\documentclass[letterpaper,11pt]{article}
\usepackage{tgpagella}
\usepackage[utf8]{inputenc}
\usepackage{graphicx,fullpage,paralist}
\usepackage{amssymb, amsthm}
\usepackage{comment,hyperref}
\usepackage{thmtools}

\usepackage{caption}
\usepackage{tikz}
\usetikzlibrary{shapes.geometric}
\usepackage{pgfplots}
\usepackage{varwidth}
\usepackage{geometry}
\usepackage{dsfont}
\usepackage{enumitem}
\usepackage{amsmath}
\usepackage{multirow,array}
\usepackage{caption}
\usepackage{subcaption}
\pgfplotsset{compat=1.10}
\usepgfplotslibrary{fillbetween}
\usetikzlibrary{backgrounds}
\usetikzlibrary{patterns}
\newcommand{\PP}{\mathbb{P}}
\newcommand{\RR}{\mathbb{R}}
\newcommand{\NN}{\mathbb{N}}

\newcommand{\EE}{\mathbb{E}}
\newcommand{\E}{\text{E}}
\newcommand{\PA}{\text{PoA}}
\newcommand{\PS}{\text{PoS}}
\newcommand{\PR}{\text{PR}}
\newtheorem{thm}{Theorem}

\newtheorem{lemma}{Lemma}

\newtheorem{remark}{Remark}

\def \R{\mathbb{R}}
\def \E{\mathbb{E}}

\def \indi{\mathds{1}}
\def \1{\mathds{1}}

 \def \1{{\bf 1}}

\newtheorem{definition}{Definition}
\newtheorem{example}{Example}
\newtheorem{proposition}{Proposition}

\usepackage{multirow}
\usepackage[noend]{algpseudocode}
\usepackage{algorithm,algorithmicx}

\newlength{\algofontsize}
\usepackage[textsize=small,textwidth=2cm]{todonotes}

\usepackage{xcolor}
\hypersetup{
	colorlinks,
	linkcolor={red!50!black},
	citecolor={blue!50!black},
	urlcolor={blue!80!black}
}

\begin{document}
	
	\algrenewcommand\algorithmicrequire{\textbf{Input:}}
	\algrenewcommand\algorithmicensure{\textbf{Output:}}
	
	\title{Competition and Recall in Selection Problems \footnote{\today}\footnote{The authors
gratefully acknowledge funding from ANITI ANR-3IA Artificial and Natural
Intelligence Toulouse Institute, grant ANR-19-PI3A-0004, and from the ANR under the Investments for the Future program, grant ANR-17-EURE- 0010. J. Renault also acknowledges the support of ANR MaSDOL-19-CE23-0017-01.}}

	\author{Fabien Gensbittel
	\thanks{Toulouse School of Economics, University of Toulouse Capitole, Toulouse, France,  fabien.gensbittel@tse-fr.eu}
	\and Dana Pizarro			
	\thanks{Toulouse School of Economics,  Universit\'e Toulouse 1 Capitole and ANITI, France, dana.pizarro@tse-fr.eu}
	\and J\'{e}r\^{o}me Renault 
	\thanks{Toulouse School of Economics, University of Toulouse Capitole, Toulouse, France,   jerome.renault@tse-fr.eu}
	}
\date{\vspace{-1em}}
\maketitle
\thispagestyle{empty}
\begin{abstract}
 
We extend the prophet inequality problem to a competitive setting. At every period $ t \in \{1,\ldots,n\}$, a new 
sample $X_t$ from a known distribution $F$ arrives, which is publicly observed. Then  two players   simultaneously decide whether to pick  an available value
or to pass and wait until the next period (ties are broken uniformly at random).
As soon as a player gets one sample, he leaves the market and his payoff is the value of this item.  
In a first variant, namely  ``no recall'' case, the agents can only bid in period $t$ for the current  
value $X_t$.  
In a second variant, the ``full recall'' case, the agents can also  bid at period $t$ for any of the previous samples with values $X_1$,...,$X_{t-1}$ which has not been already selected. 
For each variant, we study the subgame-perfect Nash equilibrium payoffs of the corresponding game, as a function of the number of periods $n$ and the distribution $F$.
More specifically, we give a full characterization in the full recall case, and show in particular that  both players always get the same payoff at equilibrium, whereas in the no recall case the set of equilibrium  payoffs typically has full dimension. Regarding the welfare at equilibrium,  surprisingly  it is possible  that the best equilibrium payoff a player can have is strictly higher in the no recall case than in the full recall case. However, symmetric equilibrium payoffs are  always better when the players have  full recall.
Finally, we show that in the case of 2 arrivals and arbitrary distributions on $[0,1]$, the prices of Anarchy and Stability in the no recall case are at most 4/3, and this bound is tight.\\

\end{abstract}

\clearpage 
\pagenumbering{arabic}
\newpage
\section{Introduction}

\subsection{Context}
The theory of optimal stopping has a vast history and is concerned with the problem of a decision-maker who observes a sequence of random variables arriving over time and has to decide when to stop optimizing a particular objective. 
Probably, the two best-known problems in optimal stopping are the Secretary  Problem and the Prophet Inequality. In the classical model of the former, introduced in the ‘60s, a decision-maker observes a sequence of values arriving over time and has to pick one in a take-it-or-leave-it fashion maximizing  the probability of picking the highest one. In other words, after observing an
arrival, he has to decide whether to pick this value (and gets a reward equals to the value picked) or to pass and continue observing the sequence.
Once a value is picked, the game ends and the goal of the decision-maker is to maximize the probability of getting the highest value. Lindley \cite{L61} proves that an optimal stopping rule for this problem consists in  rejecting  a particular amount of  
values first and then accepting the first value 
higher than the maximum observed so far.  When the number of arrivals goes to infinity, the probability of picking the best value approaches to $1/e$. Since then, several variants have been studied in the literature (see, e.g., \cite{EFK20, F83, IKM06}). 

Related to the secretary problem is the optimal selection problem where the decision-maker knows not only the total number of arrivals but also the distribution behind them, and he has to decide when to stop, with the goal of maximizing the expected value of what he gets. Instead of looking at this problem as an optimal stopping problem, in the ‘70s researchers started to answer the question of how good can a decision-maker play  compared to what a prophet can do, where a prophet is someone who knows all the realizations of the random variables in advance and simply picks the maximum. These inequalities are called Prophet Inequalities and it was in the '70s when Krengel and Sucheston, and Garling \cite{KS78} proved that the decision-maker can get at least 1/2 of what a prophet gets, and that this bound is tight. Later, in 1984 Samuel-Cahn \cite{SC84} proved that instead of looking  at all feasible stopping rules, it is enough to look at a single threshold strategy to get the 1/2 bound. These results are for a general setting where the random variables are independent but not necessarily identically distributed, and then one natural question that arose was if this bound could be improved assuming i.i.d. random variables. Kertz \cite{K86} answered this question positively and provided a lower bound of roughly 0.7451. Quite recently, Correa et al. \cite{CFHOV17} proved that this bound is tight. A lot of work has appeared considering different model variants (infinitely many arrivals, feasibility constraints, multi-selection, etc) but it was since the last decades that this problem gained particular attention due to its surprising connection with online mechanisms (see, e.g., \cite{CHMS10, CFPV19, HKS07}).

\subsection{Our paper}
In the i.i.d. setting of the prophet inequality problem, there is a decision-maker who observes an i.i.d. sample $X_1, \dots, X_n$ distributed according to a known distribution $F$. 
Only one value can be selected, and the banchmark is $\EE(\max_i X_i)$.
After each arrival, the decision-maker must make an immediate and irrevocable decision whether to pick this value or not. If he selects it, he leaves the market and the game ends, otherwise he continues observing the next arrival.  This problem is commonly motivated by applications to online auction theory, where at each time period a seller, who wants to sell an item, receives an offer from a potential buyer and has to decide whether to accept it and sell the item or to reject it and wait for the next offer. Most of the work on the prophet inequality problem relies in the fact that there is only one decision-maker and in that the decision must be taken immediately after observing the offer. However, in some situations of interest-- a person who wants to buy a house, a company hiring an employee, among others-- it seems reasonable to allow more than one decision-maker, as well as to be able to make the decision later on time. Driven by this fact, in this paper we study two variants of the classic setting. More specifically, on one hand, we consider the setting where two decision-makers compete to get the best possible value, and we call it \textit{competitive selection problem with no recall}. On the other hand, we consider the problem where the  two decision-makers are allowed to select any available value that appeared in the past and not only the one just arrived, and we call it \textit{competitive selection problem with full recall}.

For each variant, we study the two-player game induced by the optimal stopping problem, focusing on   subgame-perfect Nash equilibria   (SPE  for short).
The main contributions of this paper can be divided into three lines, which are summarized in what follows. \\

\noindent{\bf Description  of  the subgame-perfect equilibrium payoffs.} 
The first stream of contributions refers to the study of the set of SPE payoffs (SPEP) for both settings. Regarding the full recall case, we fully characterize the set of SPEP (Theorems~\ref{thm:SPEPcharact} and \ref{thm:caract}) and we obtain that  each such payoff is symmetric, meaning that every SPE gives the same payoff to both players. 
In the no-recall case, the set of SPEP is clearly symmetric with respect to the diagonal but contains points outside of the diagonal, inducing different payoffs for the players. In this case we give in Theorem ~\ref{thm:norecall} recursive formulas to  compute the best and the worst SPEP payoffs (in terms of sum of payoffs of the players), as well as the worst payoff a single player can get at a subgame-perfect equilibrium. 
To illustrate these results, in Section~\ref{sec:unif} we provide a detailed study of the particular case where  $F$ is the uniform distribution on $[0,1]$. \\

\noindent{\bf Comparison of the two variants.}
The second stream of results is focused on comparing the highest SPEP in both settings. Surprisingly, an example (see Section~\ref{sec:simple:example}) shows that the best SPEP a player can obtain may be higher in the no recall case than in the full recall case. This can be explained by  the possible  existence of an asymmetric SPE in the no recall case which significantly favors one player at the expense of the other, who picks a value early in the game. However, if we restrict our attention to symmetric SPEP, we show in Theorem~\ref{thm:comparison} that players are always better off in the variant with recall. More precisely, we prove that for every number of periods $n$ and continuous distribution $F$, if $(u,u)$ is a SPEP of the game with full recall and $(v,v)$ is a SPEP of the game without recall, then $u\geq v$. Furthermore, this advantage can be significant: If $F$ is the uniform distribution on $[0,1],$ and $n=5$, the payoff of players in the full recall case is at least $7\%$ higher than in the no recall case.\\

\noindent {\bf Efficiency of equilibria.} 
To analyze efficiency, we use in Section~\ref{sec:eff} the standard notions of Price of Anarchy and Price of Stability, and we introduce the notion of Prophet Ratio, defined as the sum of payoffs a team of prophets would obtain, divided by the  best  sum of payoffs at equilibrium. When $F$ is uniform on $[0,1]$, we find numerically that equilibria are quite efficient and that having full recall gives a small advantage to the players, in terms of having a payoff closer to the one obtained by playing the best possible feasible strategy in the corresponding setting. We also show that both the prices of Anarchy and Stability are maximized in the no recall case with $n=2$, and the intuition  is that it is the case maximizing the probability that a player will get nothing. The Prophet Ratio in the same context is maximized for $n=5$, which is less intuitive. 
Finally we consider the case $n=2$ without recall, and let $F$ be any distribution with support contained in [0,1]. We prove in Theorem \ref{prop:boundeff} that both the prices of anarchy and stability are not greater than 4/3, and that this 
bound (reminiscent of the bound for routing problems with linear latencies, see \cite{RT02}) is tight in both cases. \\

It is worth noting that, although competition and recalling variants have already been considered in the literature for some optimal stopping problems (as we will discuss in the next section), to the best of our knowledge, our paper is the first considering them from a Game Theory perspective with a focus on the study of the set of SPEP, and then it constitutes a good starting point for future research work of interest cross academic communities in Operations Research, Computer Science, and Economics.

\subsection{Related literature} \label{sec:rel:lit}

As it was aforementioned, the literature in optimal stopping theory is extensive and   mainly focused on finding optimal or near-optimal policies for the different model variants, as well as on studying the guarantees of some simple strategies, such as single threshold strategies, even when they are not optimal. However, this paper introduces a game-theoretic approach for a model with competition and where recall is allowed. In what follows, we revisit some of the existing literature regarding optimal stopping problems with some of these two particular features.


\paragraph{Optimal Stopping with competition.}
Abdelaziz and Krichen \cite{AK07} survey  the literature  on optimal stopping problems with more than one decision-maker until 2000s. More recently, Immorlica et al. \cite{IKM06} and Ezra et al. \cite{EFK20} study the secretary problem with competition. The former considers a classical setting where decision are made in a take-it-or-leave-it fashion and ties are broken uniformly at random, and they show that as the number of competitors grows, the moment at which ``accept'' is played for the first time in an equilibrium decreases. The latter incorporates the recalling option, and studies the structure and performance of equilibria in this game when the ties are broken uniformly at random or according to a global ranking. 

Our paper considers a different model, since the problem is more related to the prophet inequality setting. In this sense,  the work closest to ours is the recent paper by Ezra et al. \cite{EFK21}, who introduce, independently of our paper, the no recall case, again with the two variants for tie breaking. However, the novelty of our paper is the incorporation of recalling in addition to the competition between players. Moreover, instead of studying  the reward guarantees under single-threshold strategies, we focus  on the study of equilibria of the game.

\paragraph{Optimal Stopping with recall.}
Allowing decision-makers to choose between any of the values arrived so far  
is a variant of the classic problem that may have interesting applications. 
If this extension is considered without competition, it is easy to see that the optimum is just to wait until the end and pick the best value. However, adding competition in the model makes the problem interesting and characterizing the set of equilibrium payoff and studying their efficiency are challenging questions we address in this paper. 
Notice that this notion of recalling is not new for some optimal stopping problems. For example, Yang \cite{Y74} considers a variant of the secretary problem, where the interviewer is allowed to make an offer to any applicant already interviewed. In his model, the applicant reject the offer with some probability that depends on when the offer is made and he studies the optimal stopping rules in this context. Thereafter,  different authors have been studying other variants of the secretary problem with recall (see, e.g., \cite{EFK20, P81, S94}).
Our work differs from most of them not only in the model we consider but also in terms of questions since we are more interested in a game-theoretic approach of the problem. \\


\subsection{Roadmap}
The remainder of this paper is organized as follows: We start presenting the model in Section~\ref{sec:model}, including the description of the games for the full recall and no recall cases in Section~\ref{sec:model:game}. The results of our paper are presented formally in Section~\ref{sec:results}.  This section is divided into four parts: In Section~\ref{sec:payoffs}, we study the set of SPEP, starting with its computation for a very simple example (see Section~\ref{sec:simple:example}), and making then a complete analysis for the full recall case (see Section~\ref{sec:payoffs:FR}) and for the no recall case (see Section~\ref{sec:payoff:NR}). In Section~\ref{sec:comparison} we present the results concerning the comparison of the two variants, whereas in Section~\ref{sec:eff} we study the efficiency of equilibria. In Section~\ref{sec:unif} we make a detailed analysis of the case where $F$ is uniform on $[0,1]$. Finally, we close the paper with the proofs of the results, which are
relegated to Section~\ref{sec:proofs}.

\section{Model}\label{sec:model}

Consider a sequence of samples $X_1$,...,$X_n$ of a random variable $X$ distributed according to a c.d.f. $F$ with support included in [0,1].  \footnote{We assume the support of $F$ included in $[0,1]$ to not overload the notation but the results can be easily extended to the general case.}
There are two players, or decision-makers, competing to select the highest value among $X_1$,...,$X_n$.
More explicitly, at each time period $t=1,...,n$, the decision-makers -- who know the distribution $F$-- observe $X_{t}$ and simultaneously decide whether or not  to select one value in the current  {\it{feasible set}} $\mathcal{F}_t$. Once a decision-maker chooses a value $X_j$, he leaves the market obtaining a payoff of $X_{j}$ and this value no longer belongs to the feasible set. 
Decision-makers must decide when to stop, maximizing the expected value of their payoff.  If at time $t$ both agents want to select the same value, we break the tie uniformly at random. That is, each of them gets the value with probability $1/2$, and the decision-maker who gets it leaves the market, whereas the other passes to the next period. 

It remains to  specify what are the sets of feasible values, distinguishing the two cases we will consider. 
In the first one, namely {\it{full recall case}}, the set 
$\mathcal{F}_t \subset \{1, \dots, t\}$ of feasible elements at time $t$ consists of all the periods that have not yet been selected by a decision-maker, and a {\it{feasible action}} represents a probability distribution over the set $\mathcal{F}_t \cup \{ \emptyset \}$, where $\emptyset$ represents the action of not selecting anything.
In this case, the game ends when all decision-makers  choose an element, and then it could be later than $n$ if both players are still present at period $n$.
We assume that at stage $n$ all players who are present select the element corresponding to the highest value. If both players are present, there is a tie, and the player who losses the toss gets the second best value.


The second variant we consider is the {\it{no recall case}}. In this case, a take-it-or-leave-it decision is faced by the decision-makers at each time period. That is, after observing the sample just arrived they should decide whether to take it or not. 
Independently of their decision, the value cannot be chosen later. 
Thus, the set of feasible elements at time $t$ is just given by $\mathcal{F}_t= \{t\}$ and a  feasible action is a probability distribution on $\{t, \emptyset\},$ where $\emptyset$ represents the action of represents the action of not selecting anything.
Notice that in the 1-player case (decision problem), the optimal strategy with full recall is simply to wait until the end and to pick the element which corresponds to the maximum of $\{X_1,...,X_n\}$. And in the 1-player case with no recall  we have a standard prophet  problem, with value smaller than the expectation of the maximum of $\{X_1,...,X_n\}$. Obviously, having the possibility of getting a sample observed in the past is beneficial to the player.  
In our multiplayer setting,  we cannot say {\it a priori} that the full recall case is beneficial to the players, as there are examples where having more information or more actions decreases the sum of  the  payoffs of the players at equilibrium. This motivated us to ask how important is to have the power of being able to choose a value observed in the past. 
To answer this question we study the games behind the two model variants described above. In particular, we study the set of SPEP as well as the Price of Anarchy (PoA) and the Price of Stability (PoS). 

We remark here that throughout the paper we will  ``the decision-maker selects value $X_i$'' to refer to the decision-maker selecting element $i$ from $\mathcal{F}_t$.

\subsection{Description of the games} \label{sec:model:game}

We now formally describe the games induced by the full recall case and the no recall case, denoted by $\Gamma^{FR}_n$ and $\Gamma^{NR}_n$, respectively. Then, we specify the notions of equilibrium that will be used throughout the paper. 

\begin{enumerate}
    \item {\bf Full recall case.} Game $\Gamma^{FR}_n$

For each $t \in \{0,1,...,n+1\}$ we denote by $\mathcal{H}_t$ the set of possible histories up to stage $t$. $\mathcal{H}_0$ only contains the empty history. $\mathcal{H}_1$ contains what happens at stage 1, i.e., the sample $X_1$, who tried to get $X_1$ and who got $X_1$ (possibly nobody).  $\mathcal{H}_2$ contains everything that happened at stages 1 and 2,
and so on.

As usual a strategy for player $j \in \{1,2\}$ is an element $\sigma_j=(\sigma_{j,_t})_{t =0,...,n}$ where $\sigma_{j,_t}$ is a measurable map which associates to every history in  $\mathcal{H}_t$ an available action, that is a  probability distribution over $\mathcal{F}_t \cup \{ \emptyset \}$. A strategy profile $(\sigma_1, \sigma_2)$ induces a probability distribution over the set of possible plays $\mathcal{H}_{n+1}$, and the payoff (or utility) of each player is defined as the expectation of the value 
he gets, with the convention that getting 
none of the samples yields a payoff of $0$.

\item {\bf No recall case.} Game $\Gamma^{NR}_n$

Here, the set of available elements at stage $t$ is $\mathcal{F}_t =\{t\}$, and  we only need to consider histories $\mathcal{H}_t$ for $t \in \{0,...,n\}$  where $\mathcal{H}_t$  contains everything that happened up to stage $t$ under the no recall assumption. 
A strategy  for player $j \in \{1,2\}$ is an element $\sigma_j=(\sigma_{j,t})_{t =0,...,n-1}$ where $\sigma_{j,t}$ is a measurable map which associates to every history in  $\mathcal{H}_t$ an available action, that is a  probability distribution over $\{t, \emptyset\}$.
A strategy profile $(\sigma_1, \sigma_2)$ induces a probability distribution over the set of possible plays $\mathcal{H}_{n}$, and payoffs are defined as in the full recall case.

\item{\bf Equilibrium notions.} 
We recall here   the usual notions of Nash equilibrium (NE) and subgame perfect equilibrium (see, e.g., \cite{FT91}). The following  definitions  apply to both games $\Gamma^{FR}_n$  and  $\Gamma^{NR}_n$. 
\begin{definition}
A strategy profile  $\sigma= (\sigma_1,\sigma_2)$ is a Nash equilibrium (NE) of the game if for every agent $i$ and every strategy $\sigma'_i$, player's $i$ utility when $(\sigma'_i, \sigma_{-i})$ is played is not greater than the one obtained if $\sigma$ is played. 
\end{definition}

Given a stage number $t$ and a  finite history $h_t$ mentioning everything that happened  up to stage $t$, we can define the continuation game after $h_t$. 

\begin{definition}
A strategy profile  $\sigma= (\sigma_1,\sigma_2)$ is a subgame perfect equilibrium (SPE) if it induces  a NE for every proper subgame of the game (i.e. for any continuation game after a finite  history).
\end{definition}
 When studying SPE, we will assume w.l.o.g. that as soon as a player is alone in the game, he plays optimally in the remaining decision problem. By best (resp. worse) equilibrium, we mean an  equilibrium maximizing (resp. minimizing) the sum of the payoffs of the players.
 
\end{enumerate}
\section{Results} \label{sec:results}

The goal of this section is to present the main results of the paper. First, we study the structure of the sets of SPEP. We start by introducing a simple example, where the computation for both the recall and no recall cases is easy. Then, we present our main result in the full recall case which fully characterizes the set of SPEP. Regarding the no recall case, we provide recursive formulas to compute the the worst payoff a player can get in equilibrium, as well as the sum of payoffs of players for both the best and worst SPE, when the distribution function $F$ is continuous. 
After that, we move on to understand what is the relationship between the payoff a player gets in both problems. In other words, we answer the question: Is the recall case always beneficial for the players? Finally, we study the efficiency of equilibria for the full recall and no recall cases. We leave the proofs of the results to Section~\ref{sec:proofs}.

\subsection{Motivating example.}\label{sec:simple:example}
Let us consider the particular instance of the problem where $n$ samples of $X$ arrive sequentially over time, where the law $F$ of $X$ is defined by:

 \[   
X= 
     \begin{cases}
    1/3  & \text{with probability } 1/2,\\
      2/3 &\text{with probability } 1/2.\\ 
     \end{cases}
\]
We compute the set of SPEP in the full recall and no recall case, denoted by $E_n^{FR}$ and $E_n^{NR}$, respectively.\\ 

\noindent
\textit{Full recall case.} Note that in this case, the unique SPE is to bid at any time if $X=2/3$ and to pass otherwise until the last period, where the best available value is chosen. 
Then, if all samples take the value $1/3,$ the players obtain a payoff of $1/3$; if only one sample is equal to $2/3$, both players bid for this value and then they obtain $2/3$ with probability $1/2$ and $1/3$ with probability $1/2$; and if more than one realization of $X$ has value $2/3,$ both players obtain a payoff of $2/3$. Thus, the expected payoff of a player in equilibrium, namely $h_n$, can be computed as follows
\[
h_n=\frac{1}{3} \PP( X_{(n)} =1/3)+ \frac{1}{2} \PP( \exists! \ i : X_i=2/3 ) + \frac{2}{3} \PP(\exists i\neq j : X_i=X_j=2/3),
%
\]
where $X_{(n)}$ denotes the maximum of $n$ i.i.d. samples of $X \sim F$. The probabilities above are easy to compute, having that
$\PP( X_{(n)} =1/3)=1/2^n$, $\PP( \exists! \ i : X_i=2/3 )= n/2^{n}$ and $\PP(\exists i\neq j : X_i=X_j=2/3)=1-(n+1)/2^{n}$.
Putting all together we conclude that
$E_n^{FR}= \left\{\left( h_n,h_n\right) \right\},$
with $$h_n=\frac{2}{3}-\frac{n+2}{3 \cdot 2^{n+1}}.$$

\noindent
\textit{No recall case.} 
Because we are considering an instance with $n$ time periods and the random variable only takes two possible values, for the first $n-2$ arrivals, at equilibrium, players bid if and only if the sample has value $2/3$. Therefore, at equilibrium: if at least two samples have value $2/3$ up to time $n-2$, players obtain $2/3$; if only one samples have value $2/3$ up to time $n-2$, each player obtain $2/3$ with probability $1/2$ and the expected value of $X_{(2)}$ with probability $1/2$; and if all samples up to $n-2$ are $1/3$, the players obtain an expected payoff $(e_n^1,e_n^2) \in E_2^{NR}$. Note that $\EE(X_{(2)})$ is the value of the decision problem in a standard prophet setting with two arrivals from $X$, and thus  $\EE\left(X_{(2)}\right)=7/12$. Therefore, an expected payoff of player 1 can be computed as
\begin{eqnarray*}
\PP\left( X_{(n-2)}=1/3\right) e_2^1+\PP( \exists! \ i\leq n-2 : X_i=2/3 )\frac{15}{24}+ \PP(\exists i\neq j : X_i=X_j=2/3, i,j \leq n-2) \frac{2}{3},
\end{eqnarray*}
whereas for player 2 we have 
\begin{eqnarray*}
\PP\left( X_{(n-2)}=1/3\right) e_2^2+\PP( \exists! \ i\leq n-2 : X_i=2/3 )\frac{15}{24}+ \PP(\exists i\neq j : X_i=X_j=2/3, i,j \leq n-2) \frac{2}{3}.
\end{eqnarray*}
Again, the probabilities above are easy to compute, having $\PP\left( X_{(n-2)}=1/3\right)=1/2^{n-2}$, $\PP( \exists! \ i\leq n-2 : X_i=2/3 )=(n-2)/2^{n-2}$ and  $\PP(\exists i\neq j : X_i=X_j=2/3, i,j \leq n-2)=1-(n-1)/2^{n-2}$. Thus, we obtain that an expected payoff of player $i$ is given by: 
\begin{eqnarray}\label{eqn:example:NR}
 \frac{e_2^i}{2^{n-2}}+ \frac{15}{24} \cdot\frac{n-2}{2^{n-2}}+\frac{2}{3} \cdot \left(1-\frac{n-1}{2^{n-2}} \right)=\frac{2}{3}+\frac{e_2^i}{2^{n-2}}-\frac{n+4}{3\cdot 2^{n+1}}.
\end{eqnarray}

It remains to compute the set $E_2^{NR}.$ To this end, for  $a\ \in \{1/3,2/3\}$ we consider the game, called $\Gamma_1^{NR}(a)$, defined as $\Gamma_1^{NR}$ but with an initial available value $a$. That is, there is a time period ``zero'' where players choose between selecting the value $a$ or pass, before observing the value of the single sample of the game. We denote by $E_1^{NR}(a)$ the set of SPEP of the game $\Gamma_1^{NR}(a)$. Notice that 
\begin{equation}\label{eqn:nr:simple:example}
    E_2^{NR}= \left\{(\gamma_1,\gamma_2) : \gamma_i= \sum_{a \in  \{1/3,2/3\} } \PP(X=a) e_1^i(a), \text{where} \left(e_1^1(a), e_1^2(a) \right) \in E_1^{NR}(a)\right\},
\end{equation}
and then $E_2^{NR}$ can be easily computed from $E_1^{NR}(a).$ To study $E_1^{NR}(a),$ we consider the payoffs matrix for the game $\Gamma_1^{NR}(a),$ represented in Table~\ref{tab:payoff:example}.

  \begin{table}[h]
 \centering
    \setlength{\extrarowheight}{7pt}
    \begin{tabular}{cc|c|c|c|}
      & \multicolumn{1}{c}{} & \multicolumn{2}{c}{Player $2$}\\
      & \multicolumn{1}{c}{} & \multicolumn{1}{c}{$a$}  &  \multicolumn{1}{c}{$\emptyset$} \\\cline{3-4}
      \multirow{2}*{Player $1$}  & $a$ & $\left(\frac{1}{2} a + \frac{1}{4} , \frac{1}{2} a + \frac{1}{4}\right)$ & $\left( a, \frac{1}{2}\right)$ \\\cline{3-4}
    
       \multirow{2}*{}  & $\emptyset$ & $\left(\frac{1}{2},a\right)$ &  $\left( \frac{1}{4}, \frac{1}{4}\right)$ \\\cline{3-4}
    \end{tabular}
    
    \caption{\label{tab:payoff:example}Payoffs matrix for the game $\Gamma_1^{NR}(a)$.}
 \end{table}
 
 Observe that  if $a=2/3$, $(a,a)$ is the unique NE with payoff $(7/12, 7/12)$, and thus $E_1^{NR}(2/3)=\{(7/12, 7/12)\}$. Otherwise, if $a=1/3$, there are three NE: $(a, \emptyset), (\emptyset,a)$ and a symmetric mixed equilibrium in which both agents play $a$ with probability $1/2$ and pass with probability $1/2$. The equilibrium payoffs are $(1/3,1/2), (1/2,1/3)$ and $(3/8,3/8)$, respectively, and then $E_1^{NR}(1/3)=\{(1/3,1/2), (1/2,1/3), (3/8,3/8)\}.$ 
  We now can compute the set $E_2^{NR}$ by using \eqref{eqn:nr:simple:example}, obtaining 
 \[E_2^{NR}= \{(11/24,13/24), (13/24, 11/24), (23/48,23/48)\}.\]

 Finally, we go back to equations \eqref{eqn:example:NR} and we conclude that $E_n^{NR}$ is the three-elements set
 \begin{eqnarray}\label{set:SPE:exaple:NR}
 E_n^{NR}=\left\{ \left(p_n,q_n\right), \left(r_n,r_n \right),\left(q_n, r_n \right) \right\},
 \end{eqnarray} 
     with \[ p_n:=\frac{2}{3}-\frac{n+1}{3 \cdot 2^{n+1}}, \ \ q_n:=\frac{2}{3}- \frac{n+3}{3 \cdot 2^{n+1}} \ \ \text{  and  } \ \  r_n:=\frac{2}{3}-\frac{n+5/2}{3 \cdot 2^{n+1}}.\]
     Note that $p_n$ and $q_n$ represent the best and worst possible payoff of a player in equilibrium, respectively, whereas $r_n$ represents the payoff in between. 
     
     \vspace{0.5cm}
     
     Recall that $h_n=2/3-(n+2)/(3 \cdot 2^{n+1})$ and therefore     $p_n > h_n > r_n > q_n.$
  Moreover $p_n+q_n=2h_n>2r_n$. The latter means that the sum of equilibrium payoffs of the players in the full recall case is equal to the sum of the payoffs of players in the no recall case if an asymmetric equilibrium is played. Otherwise, i.e., if in the no recall case the symmetric equilibrium is played, then the sum of payoffs of the players is strictly lower than in the full recall case,  which is somehow not surprising. 
  
  However, for this example we observe that, surprisingly, the best payoff a single player can obtain at equilibrium is strictly  higher in the no recall case than in the full recall case.  Intuitively, what happens here is that the lack of recalling power implies that with positive probability one player will take the ``bad'' value  in period $n-1$ and therefore it gives an advantage to the other player, who is now alone in the game.  
  
A natural question is to ask how efficient are the equilibria of the games compared with the best feasible strategy. We use the well known notions of Price of Stability (PoS) and Price of Anarchy (PoA) which are defined as the ratio between the maximal sum of payoffs obtained by players under any feasible strategy and the best and worst sum of payoffs of SPEs respectively.  For the example, this analysis is easy to do, and we start by noting that in both model variants, there exist a feasible strategy (not the same for both problems) where players obtain the first and the second best values. This implies that the numerator of PoA and PoS will be the same and equal to $\EE(X_{(n)}+X_{(2:n)})$, where $X_{(2:n)}$ denotes the second best sample, which is given by $4/3-(n+2)/(3 \cdot 2^n ).$ Observe that this value is exactly $2 h_n$ and $p_n+q_n$ and then PoS is 1 for both settings and PoA is 1 for the full recall case, meaning that in the full recall case the resulting equilibria is efficient. On the other hand, the PoA for the no recall case is given by $\EE(X_{(n)}+X_{(2:n)})/(2 r_n)$ which is strictly greater to 1 but it converges fast to one when $n$ grows.

This example highlights the relevance of our work in three ways. First, even in this case where the random variable only takes two possible values, we observe that the computation of the SPEP requires a detailed analysis and then computing these sets for a general probability distribution, even under some mild assumptions, seems difficult and challenging. On the other hand, this example shows us that the power of recalling is not always favorable to all players, and then it is an interesting question to try to understand under which assumptions one can ensure that the payoff of a player in the full recall case in equilibria will be at least as good as the one in the no recall case. Finally, we see in this example that the SPEP are efficient and it motivates us to study if that still holds for a general distribution.

\subsection{ Description  of  the subgame-perfect equilibrium payoffs} \label{sec:payoffs}
Here we  will fully  characterize the  set of SPEP for the full recall case. For the no recall case, we provide recursive formulas to compute the best and worst sum of SPEP.

\subsubsection {Full recall case}\label{sec:payoffs:FR}

We go back to the general case and study here the game with  full recall, that is, at time $t$, any of the values so far arrived that has not been selected before can be selected.

We introduce a two-player game $\Gamma^{FR}_n(a,b)$, where for each natural number $n$ and $1 \geq a \geq b \geq 0$,  $\Gamma^{FR}_n(a,b)$ is defined as $\Gamma^{FR}_n$  with $n$ samples to arrive and $a$ and $b$  two available values present at the beginning of the game. That is, we have a time period ``zero'' where players choose between getting $a,b$ or pass, before the sequential arrival of the samples.
We denote by $E_n^{FR}(a,b) \subset \RR_+^2$ the set of the SPEP of the game $\Gamma^{FR}_n(a,b).$ Note that the set of the SPEP of the game $\Gamma^{FR}_n$ is just the set $E_n^{FR}(0,0)$.
 

The next theorem states that the set of SPEP for the defined auxiliary game $\Gamma^{FR}_n(a,b)$ is contained in the   diagonal, i.e., at each SPE both players get the same payoff. We  see this  as a surprising result. 

 \begin{thm}\label{thm:SPEPcharact}
Consider  an instance of the game $\Gamma^{FR}_n(a,b)$, for $a,b$ real numbers such that $0\leq b \leq a\leq 1$ and $n \in \NN$. 
The set of SPE payoffs is contained in the   diagonal, that is $E_n^{FR}(a,b) \subset \{(u,u), u\in [0, 1]\}$. 

 Furthermore, if we define $\text{P}E_n^{FR}(a,b)$ the projection of $E_n^{FR}(a,b)$ to $\RR$, we have  $\min~\text{P}E_n^{FR}(a,b)=l_n(a,b)$ and  $\max \text{P}E_n^{FR}(a,b)=h_n(a,b)$, where $l_n$ and $h_n$ are defined recursively as follows: 
 \begin{itemize}
     \item[i)]$l_0(a,b)=h_0(a,b)=\frac{a+b}{2};$
     \item[ii)] for $n\geq 1$: 
     \[l_n(a,b)=L(a,\EE\left(X_{(n)} \vee b), d^-_n(a,b) \right) \text{ with  } d^-_n(a,b)= \EE_X\left( l_{n-1}(a \vee X, \text{med}[a,b,X])\right),\]
\[h_n(a,b)=H(a,\EE\left(X_{(n)} \vee b), d_n^+(a,b) \right) \text{ with  } d_n^+(a,b)= \EE_X\left( h_{n-1}(a \vee X, \text{med}[a,b,X])\right),\]
where $X_{(n)}=\max\{X_1, \dots, X_n\}$, med denotes the median, and $L: \RR^3 \rightarrow \RR $ and $H:\RR^3 \rightarrow \RR$ are defined by:
  
  \[   
L(x,y,z) = 
     \begin{cases}
       z & \text{if } x \leq y ,\\
      \frac{1}{2}(x+y) &\text{if } x > y,\\ 
     \end{cases}
\;\; {\it and}\;\;  
H(x,y,z) = 
     \begin{cases}
      \frac{1}{2}(x+y) & \text{if } x > y \vee z,\\
      z  &\text{if }  x  \leq y \vee z.\\ 
     \end{cases}
\]
 \end{itemize}

 \end{thm}

Using the theorem above, we can prove the  following result which  fully characterizes the set of SPEP of the games $\Gamma^{FR}_n$ when $F$ is continuous (i.e. when the corresponding distribution is atomless). 

 \begin{thm}\label{thm:caract}
Assume $F$ is continuous. Then, for $n\geq 1$, the set $E^{FR}_n$ of SPE payoffs of the game $\Gamma^{FR}_n$ is the segment: 
 \[E^{FR}_n=\{\lambda (l_n,l_n)+(1-\lambda) (h_n,h_n), \lambda \in [0,1]\}, \text{ where  }  l_n=l_n(0,0) \text{ and } h_n=h_n(0,0).\]
 
That is, $E^{FR}_n$  is convex, contained in the   diagonal, and its extreme points are   $(l_n(0,0),l_n(0,0))$ and $(h_n(0,0),h_n(0,0))$, where $l_n(0,0)$ and  $h_n(0,0)$  are defined in the  statement of  Theorem~\ref{thm:SPEPcharact}.  \end{thm}
 

\subsubsection{No recall case}\label{sec:payoff:NR}
We now consider the no recall variant, where players can only play in a take-it-or-leave-it fashion, without being able to select a sample arrived in the past. 
In this section, we assume that $F$ is continuous, which ensures that the set of SPEP is convex, and allows us to derive explicit recursive formulas for the support  function of this set in particular directions.

We  introduce here  the two-player  game $\Gamma^{NR}_n(a)$, where for each natural number $n$ and $a \in [0, 1]$, $\Gamma^{NR}_n(a)$ is defined as $\Gamma_n^{NR}$   with $n$ samples to arrive, but with   $a$ 
an available value present at the beginning. 
That is, we have a time period ``zero'' where players choose between getting $a$ or pass, before the sequential arrival of the samples.

Calling $E_n^{NR}(a) \subset \RR_{+}^2$ the set of SPEP of the game $\Gamma^{NR}_n(a)$, we have that the set of the SPEP of $\Gamma^{NR}_n$   is just $E^{NR}_n:=E_n^{NR}(0)=\EE_{a \sim F} (E_{n-1}^{NR}(a)),$ where 
\begin{eqnarray*}
\EE_{a \sim F} (E_{n-1}^{NR}(a))=  \left\{ \int_0^1 f(a) \mathrm{d}F(a), f:[0,1] \rightarrow \R^2 \text{   measurable with } f(a)\in E_{n-1}^{NR}(a) \text{ for each } a \right\}.
\end{eqnarray*}

Below we present a technical result, stating that the set of SPE payoffs for the no recall case is symmetric with respect to the diagonal, and when $F$ is continuous, it is also convex and compact. 

\begin{proposition}\label{prop:convexity}
For each natural number $n$, the set of the SPE payoffs $E^{NR}_n$ is is symmetric with respect to the diagonal. If $F$ is continuous, $E^{NR}_n$ is convex compact. 
\end{proposition}

Although  the set of SPEP for the no recall case is convex and symmetric with respect to the diagonal, it will not be a subset of the diagonal. The recursive structure of the SPEP in the no recall case is more complex than the full recall case. However, in the main result of this section we give explicit recursive formulas to compute the sum of the SPE payoffs for the best and worst equilibria under the no recall setting. Recall that by best (resp. worst) equilibrium we mean a SPE which maximizes (minimizes) the sum of the payoffs of the  2 players. 

Before presenting the result, we introduce some necessary notation. 
We  first define by induction (with $X \sim F$):
  $$c_1=\EE(X), \; \; {\rm and}\; \forall n>1, \; c_n=\EE(X \vee c_{n-1}) .$$
  Note that $c_n$ is the value of the decision problem in a standard prophet setting. 
  
On the other hand, we denote by $\alpha_n$ ($\beta_n$)   the smallest (highest) coordinate value of a point on $E_n^{NR}$ belonging to the diagonal, and by $\alpha'_n$ the smallest coordinate of a point belonging to $E_n^{NR}$. That is,
 $\alpha_n:= \min \{ x : (x,x) \in E^{NR}_n\}$, $\beta_n:= \max \{x : (x,x) \in E^{NR}_n \}$ and $\alpha'_n= \min \{ \min \{ x,y \} : (x,y) \in E^{NR}_n\}$ (see Figure~\ref{fig:def}). It is easy to see that: $\alpha'_1=\alpha_1=\beta_1= 1/2 \cdot \E(X).$

  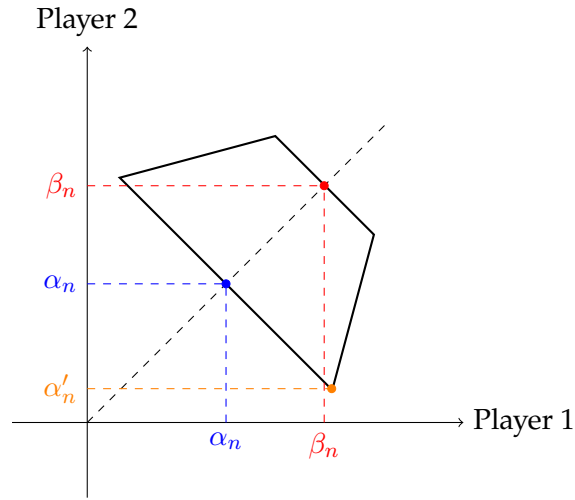
\begin{figure}[h]
 \centering 
 \begin{tikzpicture}
\draw[->] (-1, 0) -- (5, 0) node[right] {Player 1};
  \draw[->] (0, -1) -- (0, 5) node[above] {Player 2};
  \node at (2.5,2.5) (a)[trapezium, rotate=-45, trapezium angle=60, minimum width=40mm, draw, thick] {};
  \draw[scale=1, domain=0:4, dashed, smooth, name path=line, variable=\x] plot ({\x}, {\x});
\filldraw[blue] (1.845,1.845) circle (1.5pt) node[anchor=west] {};
\draw[dashed, blue] (1.845, 0) -- (1.845, 1.845) ;
\draw[dashed, blue] (0, 1.845) -- (1.845, 1.845) ;
\node[left, blue] at (0, 1.845) {$\alpha_n$};
\node[below, blue] at (1.845, 0) {$\alpha_n$};

\filldraw[red] (3.15,3.15) circle (1.5pt) node[anchor=west] {};
\draw[dashed, red] (3.15, 0) -- (3.15, 3.15); 
\draw[dashed, red] (0, 3.15) -- (3.15, 3.15);
\node[left, red] at (0, 3.15) {$\beta_n$};
\node[below, red] at (3.15, 0) {$\beta_n$};

\filldraw[orange] (3.25,0.45) circle (1.5pt) node[anchor=west] {};
\draw[dashed, orange] (0, 0.45) -- (3.25, 0.45);
\node[left, orange] at (0, 0.45) {$\alpha'_n$};
\end{tikzpicture}
 \caption{Representation of $\alpha_n$, $\beta_n$ and $\alpha'_n$  if the set $E_n^{NR}$  is given by the trapezoid.}%
   \label{fig:def}%
\end{figure}

\begin{thm}\label{thm:norecall} Assume $F$ is continuous. In the game $\Gamma^{NR}_n$, the following holds:
 
\begin{itemize}
\small
\item[a)] the {\bf worst payoff} a player can get at equilibrium  is $\alpha'_n$, where for $n\geq1$:
$$\alpha'_{n+1}=\frac{c_n+1}{2}-\int_{\alpha'_n}^{c_n} F(a) \mathrm{d}a -\frac{1}{2}\int_{c_n}^{1} F(a) \mathrm{d}a.$$

\item [b)] the sum of payoffs for the {\bf best SPE} is $2 \beta_n$, where for $n>1$:
$$2\beta_{n}=\int_{0}^{\alpha'_{n-1}} 2 \beta_{n-1} \mathrm{d}F(a) +\int_{\alpha'_{n-1}}^{\beta_{n-1}}  \max\{a+c_{n-1}, 2\beta_{n-1} \} \mathrm{d}F(a) +\int_{\beta_{n-1}}^{1} (a+c_{n-1}) \mathrm{d}F(a).$$

\item [c)] the sum of payoffs for the {\bf worst SPE}  is $2 \alpha_n$, where for $n>1$:
 \begin{eqnarray*}
 2\alpha_n&=& \int_{c_{n-1}}^{1} (a+c_{n-1})\mathrm{d}F(a) + \int_{\beta_{n-1}}^{c_{n-1}} \frac{4ac_{n-1}-2\beta_{n-1}(a+c_{n-1})}{c_{n-1}+a-2 \beta_{n-1}}\mathrm{d}F(a) + \int_{\alpha_{n-1}}^{\beta_{n-1}} 2 a \mathrm{d}F(a) \\ 
&+& \int_{\alpha_{n-1}'}^{\alpha_{n-1}}\min \{2 \alpha_{n-1}, a+c_{n-1}\}\mathrm{d}F(a) +  \int_{0}^{\alpha_{n-1}'}  2 \alpha_{n-1} \mathrm{d}F(a).
 \end{eqnarray*}
\end{itemize}
\end{thm}

\subsection{Comparison of the two model variants} \label{sec:comparison}

We  now come back to the general case and no longer assume that $F$ is continuous, and want to compare the SPEP obtained by the  players   with and without recalling power.

We have that the best SPE payoff of a player under full recall is $h_n=h_n(0,0)$ defined recursively in Theorem  \ref{thm:SPEPcharact}. We now denote   by   $$\beta'_n=\max\{\max \{x,y\}: (x,y) \in E_n^{NR}\}$$
the best possible  SPEP of a player under no-recall. 

We know from the simple example presented in Section~\ref{sec:simple:example} that we may have: $\beta'_n>h_n$, that is the best possible SPE payoff for  a player may be strcitly higher under no recall than under full recall. Unsurprisingly this is not a general result, and we will give here one example when $\beta'_2<h_2$ and one example when $\beta'_2=h_2$.


Then, we will consider the maximal sum of the payoffs of the players at equilibrium. In the the full recall case, we know that the best SPE payoff is $(h_n,h_n)$ and the worst SPE payoff is $(l_n,l_n)$. In the no-recall case when $F$ is continuous, the SPE payoff set $E_n^{NR}$ is convex, compact and symmetric with respect to the diagonal, hence the best sum of payoffs is obtained at the symmetric equilibrium $(\beta_n,\beta_n)$. Theorem \ref{thm:comparison} will show that $l_n \geq \beta_n$, implying that when $F$ is continuous, at a symmetric equilibrium payoff, players are always better off under full recall than under no recall.

\begin{example} \rm An example with $h_2 > \beta'_2$. Consider the following discrete random variable: 

 \[   
X= 
     \begin{cases}
    1/10  & \text{with probability } 1/2,\\
     1/2 &\text{with probability } 1/2.\\ 
     \end{cases}
\]

We will show that the expected payoff of a player at equilibrium in the full recall case is always higher than the expected payoff of a player at equilibrium in the no recall case, when there are 2 samples of $X$ arriving sequentially. In other words, we will prove that $h_2>\beta'_2$.

Let us start with the full recall case. In this setting, the only SPE is to bid in the first stage if and only if $X=1/2$. Then, the expected payoff of a player playing SPE is given by:
\[
h_2=\frac{1}{10} \PP(X_{(2)}=1/10)+\left(\frac{1}{2} \cdot \frac{1}{2}+\frac{1}{2} \cdot \frac{1}{10}\right)\PP( \exists! i: X_i=1/10)+\frac{1}{2} \PP(X_{(2:2)}=1/2)=\frac{3}{10}.
\]

In the no recall case, Table~\ref{tab:payoff:example1} represents the matrix of expected payoff for the game if the first arrival is $x$.
  \begin{table}[h]
 \centering
    \setlength{\extrarowheight}{7pt}
    \begin{tabular}{cc|c|c|c|}
      & \multicolumn{1}{c}{} & \multicolumn{2}{c}{Player $2$}\\
      & \multicolumn{1}{c}{} & \multicolumn{1}{c}{$x$}  &  \multicolumn{1}{c}{$\emptyset$} \\\cline{3-4}
      \multirow{2}*{Player $1$}  & $x$ & $\left(\frac{1}{2} x + \frac{3}{20} , \frac{1}{2} x + \frac{3}{20}\right)$ & $\left( x, \frac{3}{10}\right)$ \\\cline{3-4}
    
       \multirow{2}*{}  & $\emptyset$ & $\left(\frac{3}{10},x\right)$ &  $\left( \frac{3}{20}, \frac{3}{20}\right)$ \\\cline{3-4}
    \end{tabular}
    
    \caption{\label{tab:payoff:example1}Payoffs matrix for the game with $n=2$ if $X_1=x$.}
 \end{table}
Note that if $x=1/10$, $(\emptyset, \emptyset)$ is the unique NE with payoff $(3/20,3/20)$, and if $x=1/10$, $(x,x)$ is the unique NE with payoff $(2/5,2/5)$. Then, the SPE payoff of each player is $1/2  \cdot 3/20 + 1/2 \cdot 4/20=11/40.$ This means that $\beta'_2=11/40$, and we then conclude that $h_2 > \beta'_2$.
 \end{example}

\begin{example} \rm An example with $h_2 = \beta'_2$. Let us now consider the game with two samples of $X$ arriving sequentially over time, with $X \sim \text{Unif}[0,1]$. 

Using Theorem~\ref{thm:SPEPcharact} we can easily obtain that $h_2=\EE(X)$, and then in this case we have $h_2=1/2$.

Regarding the no recall case, Table~\ref{tab:payoff:example} represents the expected payoffs matrix if the first arrival is $x$. We now compute the set of NE depending on the value of $x$:

\begin{itemize}
    \item[\textit{Case 1.}] If $x<1/4$, $(\emptyset,\emptyset)$ is the unique NE and the payoff is given by $(1/4,1/4)$.
    \item[\textit{Case 2.}] If $x>1/2$, $(x,x)$ is the unique NE with payoff $(1/2x+1/4, 1/2x+1/4)$.
    \item[\textit{Case 3.}] If $x \in [1/4,1/2]$, $(x,\emptyset)$ and $(\emptyset,x)$ are the pure NE with payoff $(x,1/2)$ and $(1/2,x)$, respectively, and there is also a mixed NE with payoff $(3/4-1/(8x),3/4-1/(8x))$.
\end{itemize}

Therefore we have that,

  \[   
E_1^{NR}(x) = 
     \begin{cases}
     \{(1/4,1/4) \}& \text{if } x < 1/4\\
       \{(x,1/2),(1/2,x), (3/4-1/(8x),3/4-1/(8x)) \}  &\text{if }  x \in [1/4,1/2] \\
        \{(1/2x+1/4, 1/2x+1/4)  \}  &\text{if } x > 1/2. \\
     \end{cases}
\]

Note that the maximum possible expected payoff for one player is obtained if he passes for every $x \in [1/4,1/2]$, and the payoff obtained is given by
\[\beta'_2= \int_0^{1/4} \frac{1}{4} \mathrm{d}x +\int_{1/4}^{1/2} \frac{1}{2} \mathrm{d}x + \int_{1/2}^1 \frac{x}{2}+\frac{1}{4}  \mathrm{d}x =\frac{1}{2},\]

concluding that $h_2=\beta'_2$.
\end{example}

 \noindent{\bf Symmetric equilibrium payoffs.}
 From the examples above, we conclude that it is not true that at  equilibrium players always take advantage of the recalling power. However, restricting the set of SPE in the no recall case to the symmetric SPE, holds that the expected revenue of a player in the full recall case is always at least the one obtained in the no recall case. We state this result formally in Theorem~\ref{thm:comparison}.
 
 \begin{thm}\label{thm:comparison} Assume that $F$ is continuous. Let $(u,u)\in E_n^{FR}$ be a SPE payoff under full recall, and $(v,v)\in E_n^{NR}$ be a {\it symmetric} SPE payoff under no recall. Then $u\geq v$.
 \end{thm}
 
 The proof in Appendix~\ref{app:comparison} will use the following technical lemma,  which gives a lower bound for a SPEP of a player in the full recall case. 
 
 \begin{lemma}\label{lem:bound:SPEP}
 Let $a$ be a positive real number such that $a \leq \EE\left(X_{(n+1)}\right)$. If $\gamma_n^{FR}$ denotes the expected payoff of one player in some SPE in $\Gamma_n^{FR}(a \vee X, a \wedge X)$, then 
 \[
 \gamma_n^{FR} \geq \frac{a+c_{n+1}}{2},
 \]
 where $c_n$ is the value of the decision problem in the no recall case with one decision-maker and $n$ arrivals.
 \end{lemma}



 
 
 

\subsection{Efficiency of equilibria}\label{sec:eff}

The goal of this section is to study how efficient are the SPE  payoffs. To this end, we define as usual the Price of Anarchy and Price of Stability, and we introduce what we call the Prophet Ratio of the game. Given an instance of a game, the  first two notions refers to the ratio between the maximal sum of payoffs obtained by players under any feasible strategy and the sum of payoffs for the worst and best SPEs, respectively. 
On the other hand, we define the Prophet Ratio of an instance of the problem as the ratio between the optimal Prophet value of the problem (that is, the expected sum   of the two best values) and the sum of payoffs for the best SPE. We call this quantity Prophet Ratio because we are comparing the best sum of payoffs obtained by playing a SPE strategy with what a prophet would  do if he knows all the values of the samples in advance. 

Next, we formally introduce these definitions.

\begin{definition}
Consider an instance $\Gamma^{FR}_n$ or $\Gamma^{NR}_n$, where $n$ samples from a distribution $F$ arrive sequentially over time. Denote by $\Sigma$  the set of all feasible strategy pairs and by SPE    the set of subgame-perfect equilibria. We call:
\begin{itemize}
    \item[a)] Price of Anarchy of this game instance-- and we denote it by $\PA_n(F)$--  to the following ratio
     \[
     \PA_n(F):= \frac{ \max_{ \sigma \in \Sigma}\gamma^1(\sigma)+\gamma^2(\sigma)}{\min_{\sigma \; SPE} \gamma^1(\sigma)+\gamma^2(\sigma)},
     \]
       \item[b)] Price of Stability of this game instance-- and we denote it by $\PS_n(F)$-- to the following ratio
     \[
     \PS_n(F):= \frac{ \max_{ \sigma \in \Sigma}\gamma^1(\sigma)+\gamma^2(\sigma)}{\max_{\sigma \; SPE} \gamma^1(\sigma)+\gamma^2(\sigma)},
     \]
    \item[c)] Prophet Ratio of this game instance-- and we denote it by $\PR_n(F)$-- to the following ratio
     \[
     \PR_n(F):= \frac{\EE(X_{(1:n)}+X_{(2:n)})}{\max_{\sigma \; SPE} \gamma^1(\sigma)+\gamma^2(\sigma)},
     \]
     where $X_{(1:n)}$ and $X_{(2:n)}$ represent the first and second order statistics from the sequence of samples $\{X_i\}_{i \in [n]}$.
\end{itemize}
\end{definition}

Clearly, by definition it holds that for each $n$ and $F$
\[
\min\{\PA_n(F),\PR_n(F)\}\geq \PS_n(F)  \geq 1. 
\]

Notice that $\PS_n(F)/\PR_n(F)= \frac{ \max_{ \sigma \in \Sigma}\gamma^1(\sigma)+\gamma^2(\sigma)}{\EE(X_{(1:n)}+X_{(2:n)})}$  is usually  called competitive ratio for the two-sample optimal selection problem. 

For each $n$, we define the Price of Anarchy, Price of Stability and Prophet Ratio of our competitive selection problems  as the worst case ratio over all possible value distributions $F$. That is:
$$\PA_n:= \max_F \PA_n(F), \hspace{0.4cm} \PS_n:= \max_F \PS_n(F),\hspace{0.4cm} \PR_n:= \max_F \PR_n(F).$$

In what follows, we study this quantities for each of the model variants, and at the end of the section, we consider the case where the number of arrivals is two, and we present a tight bound for the ratios in Theorem~\ref{prop:boundeff}.

\subsubsection{Full recall case } Note that in this case, the maximal feasible sum of payoffs obtained by the players   is simply  the expected value of the two best samples (they can wait until the end of the horizon and pick the best and second best values). Thus, here, the notions of Price of Stability and Prophet Ratio are equivalent.

By Theorem~\ref{thm:caract}, the sum of payoffs for the worst SPE is given by $2l_n$ and for the best SPE is $2h_n$, and therefore given an instance of the game we have   in this setting:
\[
\PA^{FR}_n(F)= \frac{\EE(X_{(1:n)}+X_{(2:n)})}{2 l_n} 
\; \text{ and }\;
 \PS^{FR}_n(F)=\PR^{FR}_n(F)=  \frac{\EE(X_{(1:n)}+X_{(2:n)})}{2 h_n}.\]

Notice that if $n=2$, that is, we have only two arrivals, each sample will be eventually picked by a player, so   that $\PA^{FR}_2(F)=\PS^{FR}S_2(F)=\PR^{FR}_2(F)=1$ for every value distribution $F$. 

On the other hand, if $n$ goes to infinity, then we also have that both the Price of Anarchy and Price of Stability goes to 1. Then, the interesting question is what happen with these ratios when $n$ is finite and greater than 2. 

Although we have a general characterization of the ratios for any value of $n$ and distribution $F$, these quantities are not always easy  to compute for any $n$ even if we fix the distribution $F$.









\subsubsection{No recall case} 
If $n=2$, picking the two best samples is a feasible strategy since player one can get $X_1$ and player two $X_2$, and then $\PS_2^{NR}(F)=\PR_2^{NR}(F)$. However, for $n\geq 3$, as soon as the support of $F$ has at least three points, picking almost surely the two best samples is no longer feasible and $\PR_n^{NR}(F)>\PS_n^{NR}(F).$

Note that in this case, the maximal feasible strategy is the same as the strategy of one player selecting two values among $n$ in the classical online selection problem. The following Lemma gives us a recursive formula for the expected sum of payoffs of the maximal feasible strategy.

\begin{lemma}\label{lem:maxfeasible}
Assume we are in the no recall case with $n$ arrivals following a distribution $F$ with mean $m$. Let $X$ denote a random variable with law $F$, and $c_n$ the value of the decision problem in the no recall case with one decision-maker and $n$ arrivals. Then, the expected maximal feasible sum of payoffs $s_n$ satisfies
\begin{itemize}
\item[a)] $s_1= m$ if $n=1$,
    \item[b)] $s_2=2 m$ if $n=2$,
    \item[c)] $s_{n}=\PP(X \geq x_{n-1}-c_{n-1}) \EE(X+c_{n-1} | X \geq s_{n-1}-c_{n-1})+ s_{n-1} \PP(X < s_{n-1}-c_{n-1})$ if $n>2.$
\end{itemize}
\end{lemma}

Using Theorem~\ref{thm:norecall} together with Lemma~\ref{lem:maxfeasible} we have:
\[
\PA^{NR}_n(F)=\frac{s_n}{2\alpha_n}, \hspace{0.6cm} \PS^{NR}_n(F)=\frac{s_n}{2\beta_n} \hspace{0.6cm} \text{and} \hspace{0.6cm} \PR^{NR}_n(F)=\frac{\EE(X_{(1:n)}+X_{(2:n)})}{2\beta_n}.
\]


If $n=1$ the ratios are equal to 1 and if we take $n$ going to infinity, we also obtain that the ratios goes to 1. Then, as in the full recall case, the interesting cases are the ones in between. 

Recall that from Theorem~\ref{thm:comparison} we have that for every $n$ and every continuous distribution $F$, $h_n  \geq\beta_n$ and therefore $\PR_n^{FR}(F) \geq \PR_n^{NR}(F)$. However, the comparison is not direct for PoA and PoS as the numerators are different in the full recall and the no recall cases.
\subsubsection{Two arrivals case}
To conclude this section, we study the efficiency of SPE when we fix  the number of arrivals at two and we look at the worst case ratios over $F$. In the case with full recall all ratios are 1 and the question is trivial, so we consider the no recall case here. In other words, we consider the game $\Gamma_2^{NR}$ and we want to study how bad it may be to play the best or worst SPE, in terms of the sum of payoffs obtained, compared with the maximal feasible  sum of payoffs. 




We obtain the following result, which states that both the Price of Anarchy and Price of Stability are upper bounded by $4/3$, and that this bound tight. 

\begin{thm} \label{prop:boundeff}
  If $n=2$, under the no recall case it holds that for every distribution $F$,
  \[
 \PS^{NR}_2(F) \leq 4/3 \hspace{0.6cm} \text{and } \hspace{0.6cm} 
 \PA^{NR}_2(F) \leq 4/3.
  \]
Furthermore, this bound is tight for both the price of stability and price of anarchy. 
\end{thm}    


 \subsection{Example: Uniform distribution} \label{sec:unif}
 
In this section, we apply the results obtained in the former sections to the particular case where the samples are from a random variable uniformly distributed on [0,1].

 \subsubsection{Computation of SPEP}
 We start by the computation of the SPEP. To this end, we divide the analysis according to whether we are under the full recall or no recall case. For the former, We compute $E_3^{FR}$ to illustrate how to apply Theorem \ref{thm:SPEPcharact}. For the latter, we implement the recursive formulas for $\alpha_n$, $\beta_n$ and $\alpha'_n$.
 
 \paragraph{Full recall case.} We now compute the set of SPEP of the game $\Gamma^{FR}_3$, when we have three samples arriving from a uniform $[0,1]$ distribution.  
 
Given $ 1 \geq a \geq b \geq 0$, by Theorem~\ref{thm:SPEPcharact} it holds that $l_0(a,b)=h_0(a,b)=(a+b)/2$ and 
\[ 
E_0^{FR}(a,b)= \left\{ \left( (a+b)/2, (a+b)/2\right) \right\}.
\]
Let us now compute a closed-form expression for $E_1^{FR}(a,b)$.

Using the notation introduced in Theorem \ref{thm:SPEPcharact}:


\begin{eqnarray*}
d_1^-(a,b)=d_1^+(a,b)=\EE\left[\frac{a+b}{2}\indi_{X< b}+\frac{a+X}{2}\indi_{X\geq b} \right]= b(a+b)/2 +(1-b) (a/2+(1+b)/4)=\frac{1}{2}a+\frac{1}{4}+\frac{1}{4}b^2.
\end{eqnarray*}


If $X \sim \text{Unif}[0,1]$ and $k \in [0,1]$, we have:
\begin{eqnarray*}
    \EE(X \vee k)&=& k \PP(X \leq k)+ \EE(X|X >k) \PP(X>k) \\
&=&  \frac{1+k^2}{2}.
\end{eqnarray*}

Using that $\frac{1+2a+b^2}{4}=\frac{1}{2}\left(a+ \frac{1+b^2}{2} \right)$, we deduce 
that 
\[ l_1(a,b)=L\left(a, \frac{1+b^2}{2},\frac{1+2a+b^2}{4} \right)=h_1(a,b)=H\left(a, \frac{1+b^2}{2},\frac{1+2a+b^2}{4} \right)=\frac{1+2a+b^2}{4}.\] 
By Theorem~\ref{thm:SPEPcharact} we obtain

  \[
  E_1^{FR}(a,b)=\left\{\left(\frac{a}{2}+\frac{1+b^2}{4},\frac{a}{2}+\frac{1+b^2}{4}\right)\right\}.\]
  
  
  
To compute $E_2^{FR}(a,b)$, we need to compute $\EE(X_{(2)}\vee b)$. The expected value of the maximum between $n$  $\text{Unif}[0,1]$  random variables and a constant $k \in [0,1]$ is given by:
\begin{eqnarray*}
    \EE(X_{(n)} \vee k)&=& k \PP(X_{(n)} \leq k)+ \EE(X_{(n)}|X_{(n)} >k) \PP(X_{(n)}>k) \\
    &=& k^{n+1}+\frac{n}{n+1}(1-k^{n+1}) \\
    &=& \frac{n+k^{n+1}}{n+1},
\end{eqnarray*}
  and therefore $\EE(X_{(2)}\vee b)=(2+b^{3})/3$.
  
  We have:
  \begin{eqnarray*}
   d_2^-(a,b)=d_2^+(a,b)=\EE_X\left(\frac{a \vee X}{2}+\frac{1+((a \wedge X) \vee b)^2}{4}\right) 
   &=&\frac{1+a^2}{2}+\frac{b^3-a^3}{6}.
  \end{eqnarray*}

We obtain:

\[ l_2(a,b)=\begin{cases} (1+a^2)/2+(b^3-a^3)/6 & \text{if } a \leq (2+b^3)/3 \\
a/2+(2+b^3)/6 & \text{if } a \> (2+b^3)/3 \end{cases}\]
\[ h_2(a,b)=\begin{cases} (1+a^2)/2+(b^3-a^3)/6 & \text{if } a \leq \max\{(2+b^3)/3, (1+a^2)/2+(b^3-a^3)/6`\} \\
a/2+(2+b^3)/6 & \text{if } a > \max\{(2+b^3)/3, (1+a^2)/2+(b^3-a^3)/6`\}\end{cases}\]

In particular, $l_2(0,0)=h_2(0,0)=1/2$ so that $E_2^{FR}=E_2^{FR}(0,0)=\{(1/2,1/2)\}$.

Now, we compute  $E_3^{FR}$. We have 
\[ l_2(a,0)=\begin{cases} (1+a^2)/2-a^3/6 & \text{if } a \leq 2/3 \\
a/2+1/3 & \text{if } a > 2/3 \end{cases}\]
\[ h_2(a,0)=\begin{cases} (1+a^2)/2-a^3/6 & \text{if } a \leq a^* \\
a/2+1/3 & \text{if } a > a^*\end{cases},\]
where $a^*$ is the unique root of $a=(1+a^2)/2-a^3/6$.
We deduce that
\[ l_3(0,0)= d_3^-(0,0)=\EE[l_2(X,0)]=\int_{0}^{2/3} ( (1+a^2)/2 - a^3/6) \mathrm{d}a+\int_{2/3}^{1} (a/2+1/3)\mathrm{d}a= \frac{607}{972}\approx 0.6244, \] 

\[h_3(0,0)=d_3^+(0,0)=\EE[h_2(X,0)]=\int_{0}^{a*} ( (1+a^2)/2-a^3/6) \mathrm{d}a+\int_{a^*}^{1} (a/2+1/3) \mathrm{d}a 
 \approx 0.6245.\]
We conclude that 
 \[
 E_3^{FR}(0,0)=\left\{ (u,u) : u \in [l_3(0,0),h_3(0,0)]\right\}.
 \]

What is played at equilibrium? In the best equilibrium, both players pick $X_1$ if and only if $X_1\geq a^*$, whereas in the worst equilibrium, both players pick $X_1$ if and only if $X_1\geq 2/3$. Competition  induces  the players to pick $X_1$ with relatively low values in the worst equilibrium, and this  decreases the sum of expected payoffs.

\paragraph{No recall case.}
Now, we turn to reduce the formulas obtained in Theorem~\ref{thm:norecall} for the particular case where the 
values are uniformly distributed in the interval $[0,1]$.
Recall that $c_1=\EE(X)$, $c_n=\EE(X \vee c_{n-1})$ for $n>1$ and $X \sim F$. To obtain the expressions of $\beta_n$ and $\alpha_n$ for the uniform case, we use the following technical result.

 \begin{lemma} \label{lem:ineq:norecall}
 If $\alpha'_n <a< \beta_n$, then $a+c_n \geq 2 \beta_n$.
 \end{lemma}
\begin{proof}
It is enough to show that $\alpha_n'+c_n \geq 2 \beta_n$.
We prove that by induction on $n$.

Notice that $E_1^{NR}=\{(1/4,1/4)\}$ and thus $\beta_1=\alpha'_1=1/4$. On the other hand, $c_1=\EE(X)=1/2,$ and putting all together we have that
\[
\alpha'_1+c_1 \geq 2 \beta_1,
\]
and then the Lemma holds for $n=1.$

Let us assume now that the inequality holds for $n$ and we prove that it also holds for $n+1,$ that is:
\[
\alpha'_{n+1}+c_{n+1} \geq 2 \beta_{n+1}.
\]

By Theorem~\ref{thm:norecall}, we have that 
 \begin{eqnarray*}
 2\beta_{n+1}&=& 2 {\alpha'_n} \beta_n  + \int_{\alpha'_n}^{\beta_n}  \max\{a+c_n, 2\beta_n \} \mathrm{d}a +(1-\beta_n)c_n + \frac{1}{2}(1-\beta_n^2) \\
 &=&2 {\alpha'_n} \beta_n +(\beta_n-\alpha_n')c_n + \frac{1}{2}(\beta_n^2-{\alpha_n'}^2)+(1-\beta_n)c_n + \frac{1}{2}(1-\beta_n^2)\\
 &=&2 {\alpha'_n} \beta_n +c_n (1-\alpha_n')+\frac{1}{2}(1-{\alpha'_n}^2) \\
 &\leq& ({\alpha_n'}+c_n)\alpha_n'+c_n (1-\alpha_n')+\frac{1}{2}(1-{\alpha'_n}^2) \\ 
 &=& c_n+ \frac{1}{2} {\alpha_n'}^2+\frac{1}{2},
 \end{eqnarray*}
 where the second equality and the inequality follow by the induction hypothesis.
 
 On the other hand, from Theorem~\ref{thm:norecall}, we obtain
 \[
 \alpha_{n+1}'=\frac{1}{2}{\alpha_n'}^2+\frac{1}{2}c_n-\frac{1}{4} c_n^2+\frac{1}{4}.
 \]
Then, $\alpha'_{n+1}+c_{n+1} \geq 2 \beta_{n+1}$ if and only if
 \[
 \frac{1}{4}(c_n+1)^2 \leq c_{n+1},
 \]
 and the last holds due to $c_{n+1}=\EE(X \vee c_n)=(1+c_n^2)/2.$
 
 Therefore, we have that $ 2 \beta_{n+1} \leq \alpha'_{n+1}+c_{n+1},$ and the proof is completed. 
\end{proof}

Using Theorem~\ref{thm:norecall} together with Lemma~\ref{lem:ineq:norecall}, we obtain that the sum of payoffs for the best SPE for $n>1$ is 
\begin{eqnarray*}
2\beta_{n}&=& \int_{0}^{\alpha'_{n-1}} 2 \beta_{n-1} \mathrm{d}a + \int_{\alpha'_{n-1}}^{1} (a+c_{n-1}) \mathrm{d}a \\
&=&2 \beta_{n-1} \alpha'_{n-1}+ c_{n-1}(1-\alpha'_{n-1})+\frac{1}{2}(1-{\alpha_{n-1}'}^2).
\end{eqnarray*}

Notice that by definition, $\beta_n \ge \alpha_n$ for all $n$, and therefore, by Lemma~\ref{lem:ineq:norecall} we conclude that if $\alpha_n' < a< \alpha_n$ then $a+c_n \geq 2 \alpha_n$. Thus, applying Theorem~\ref{thm:norecall}, the sum of payoffs for the worst SPE for $n>1$ is 
\begin{eqnarray*}
2\alpha_n=\frac{1}{2}+\frac{5}{2} c_{n-1}^2+3 \beta_{n-1}^2+c_{n-1}-6c_{n-1}\beta_{n-1}+\alpha_{n-1}^2-4(c_{n-1}-\beta_{n-1})^2 \ln(2).
\end{eqnarray*}

If the samples are uniformly distributed in the interval $[0,1]$, the recursive formulas given in Theorem \ref{thm:norecall} are easy to implement numerically,

In Table~\ref{tab:alphas:beta} we expose the values of  $\alpha'_n, \alpha_n$ and $\beta_n$ for $n$ from 1 to 4 and $F=\text{Unif}[0,1]$, whereas in Figure~\ref{fig:alp:bet} we compare the sum of payoffs for the best and worst SPE for $n$ up to 10. 

\begin{table}[ht]
\begin{center}
\begin{tabular}{|c|c|c|c|c|}
\hline 
 & $\alpha'_n$  & $\alpha_n$ & $\beta_n$ \\ 
\hline \hline
n=1 & 0.25 & 0.25 & 0.25  \\
\hline
n=2 & 0.4688 & 0.4759 & 0.4844\\ 
\hline 
n=3 & 0.5747 & 0.5803 & 0.5881 \\ 
\hline 
n=4 & 0.6419 & 0.6465 & 0.6533 \\
\hline
\end{tabular} 
\end{center}
\caption{Values of  $\alpha'_n, \alpha_n$ and $\beta_n$ for one, two, three and four arrivals and values uniformly distributed in $[0,1]$.}
\label{tab:alphas:beta}
\end{table}

 \begin{figure}[h]
 \centering
 \includegraphics[width=1.1 \textwidth]{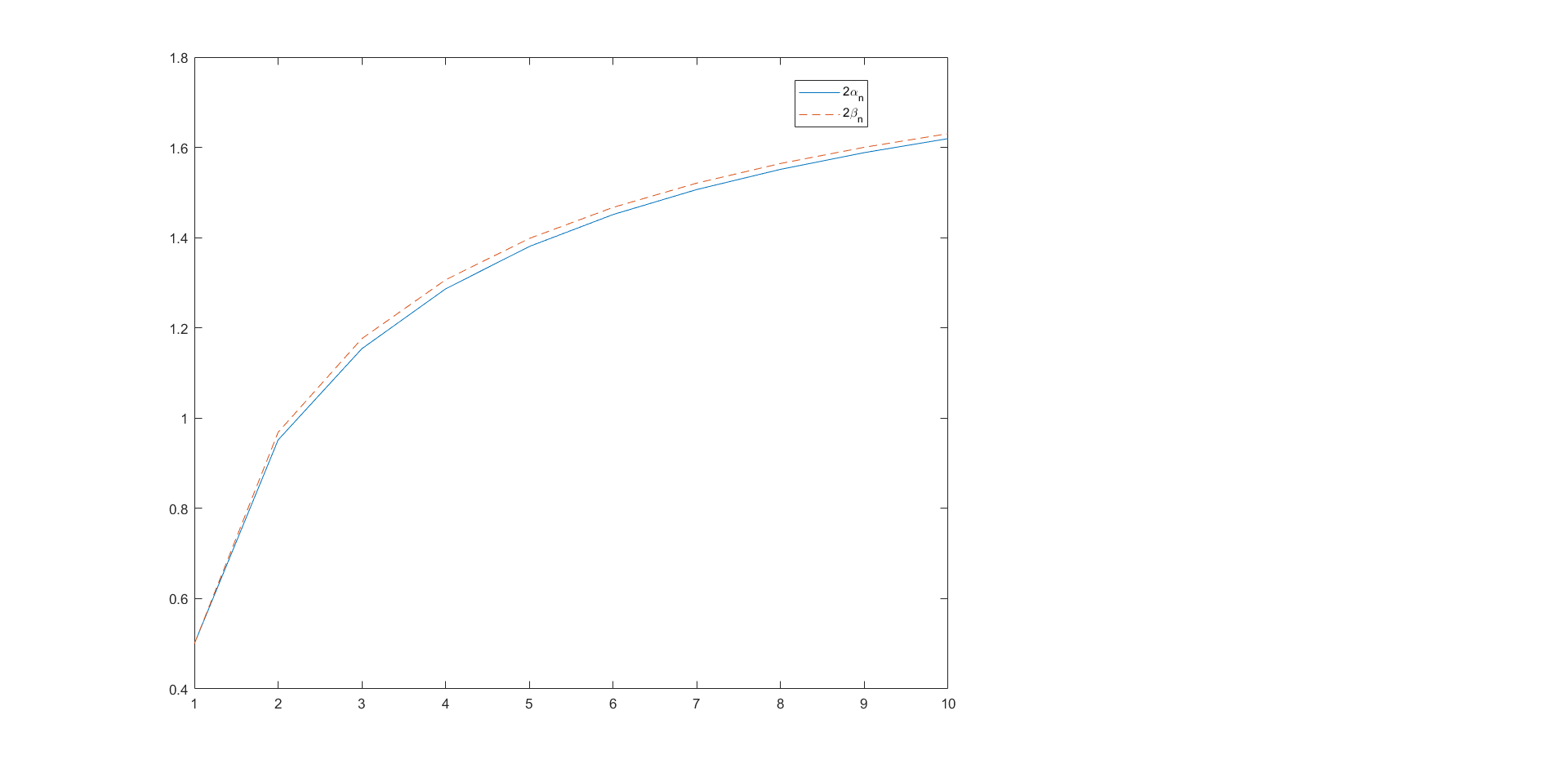}
 \caption{Values of sum of payoff for the best SPE (dashed line) and for the worst SPE (continuous line) for up to 10 arrivals, no recall  and values uniformly distributed in $[0,1]$.} \label{fig:alp:bet}
 \end{figure}

\subsubsection{Comparison of model variants}

We know by Theorem~\ref{thm:comparison} that $\alpha_n \leq \beta_n \leq l_n \leq h_n$ for every $n$, meaning that having full recall is advantageous to the players. We quantify this difference for small values of $n$ numerically, 

In Table~\ref{tab:l:h:alpha:beta} we present the values of $\alpha_n, \beta_n , l_n$ and $h_n$ for $n \leq 5$ arrivals. The values corresponding to the no recall case were computed using the formulas obtained in the former section, whereas for the full recall case, for $n=1,2,3$ the values were obtained explicitely and for $n=4,5$ numerically via discretization. 

From the values of Table~\ref{tab:l:h:alpha:beta} we can see that under the full recall case, each player has an advantage compare to the no recall case of between 3$\%$ and 5$\%$ for $n=2$, between $6\%$ and $7.6\%$ for $n=3$, between $6.9\%$ and $8\%$ for $n=4$ and between $7\%$ and $8\%$ for $n=5$. 

\begin{table}[ht]

\begin{center}
\begin{tabular}{|c|c|c|c|c|c|}
\hline 
 & $l_n$ & $h_n$ & $\alpha_n$ & $\beta_n$ \\ 
\hline  \hline 
n=1 & 0.25 & 0.25 & 0.25 & 0.25  \\
\hline
n=2 & 0.5 & 0.5 & 0.4759 & 0.4844\\ 
\hline 
n=3 & 0.6244 & 0.6245 &  0.5803 & 0.5881 \\ 
\hline 
n=4 & 0.6989 & 0.699 & 0.6465 & 0.6533 \\
\hline
n=5 & 0.7484 & 0.7486 & 0.6932 & 0.6991  \\
\hline
\end{tabular} 
\end{center}
\caption{Values of  $l_n, h_n, \alpha_n$ and $\beta_n$ for up to 5 arrivals and values uniformly distributed in $[0,1]$.}
\label{tab:l:h:alpha:beta}
\end{table}

\subsubsection{Efficiency of equilibria}

We represent in Table~\ref{tab:multicol} the values of the ratios $\PA_n(F), \PS_n(F)$ and $\PR_n(F)$ for $X \sim\text{Unif}[0,1]$ and $n$ up to 5 in both settings.  We notice that the ratios are close to 1, i.e.  equilibria are close to be efficient in both model variants, being more efficient in the full recall case. We highlight here that these ratios are a measure of efficiency of equilibria for each game, but not between the two different games. That is to say, it is not correct to say that one setting is better than the other by comparing the PoA or PoS, because for each game, these ratios are computed in a different way (the numerator are different). For such a comparison, we could use the value of PR, but comparing the PR for both problems is equivalent to compare $\beta_n$ and $h_n$, which is a comparison we already made. 


\begin{table}[ht]

\begin{center}
\begin{tabular}{|c|c|c|c|c|c|c|}
\hline 
 & \multicolumn{2}{c|}{PoA$_n$(F)} & \multicolumn{2}{c|}{PoS$_n$(F)} & \multicolumn{2}{c|}{PR$_n$(F)}\\ 
\hline \hline 
 & Full Recall  & No Recall & Full Recall  & No Recall & Full Recall  & No Recall\\ 
\hline 
n=2 & 1 & 1.0507 & 1 & 1.0323 & 1 & 1.0323 \\ 
\hline 
n=3 & 1.000823 & 1.0299 & 1.0008 & 1.0161 & 1.0008 & 1.0627\\ 
\hline 
n=4 & 1.00157 & 1.0212 & 1.00143 & 1.0105 & 1.00143 & 1.0714 \\ 
\hline 
n=5 & 1.0021 & 1.0164 & 1.00187 & 1.0077 & 1.00187 & 1.0728 \\ 
\hline 
\end{tabular} 
\end{center}
\caption{Efficiency of equilibria up to 5 arrivals from Unif[0,1] distribution.}
\label{tab:multicol}

\end{table}


Another question we address here is the number of arrivals that gives the worst gaps,   in the no recall case for which the numerics are easier. We obtain  that, for both the Price of Anarchy and Price of Stability, the ratios reach their maximum when $n=2$, see Figures~\ref{fig:PoA} and \ref{fig:PoS}.  This result is somehow intuitive: as we are in the no recall case, there exists a positive probability of getting nothing and then the smaller the number of arrivals, the more likely this seems to happen. 
Regarding the Prophet Ratio, we also compute it as a function of $n$  (see Figure~\ref{fig:PR}), and we obtained that the maximum is reached when $n=5$. Here, the result is more surprising.


\begin{figure}[h]
     \centering
     \begin{subfigure}[b]{0.45\textwidth}
         \centering
         \includegraphics[width=\textwidth]{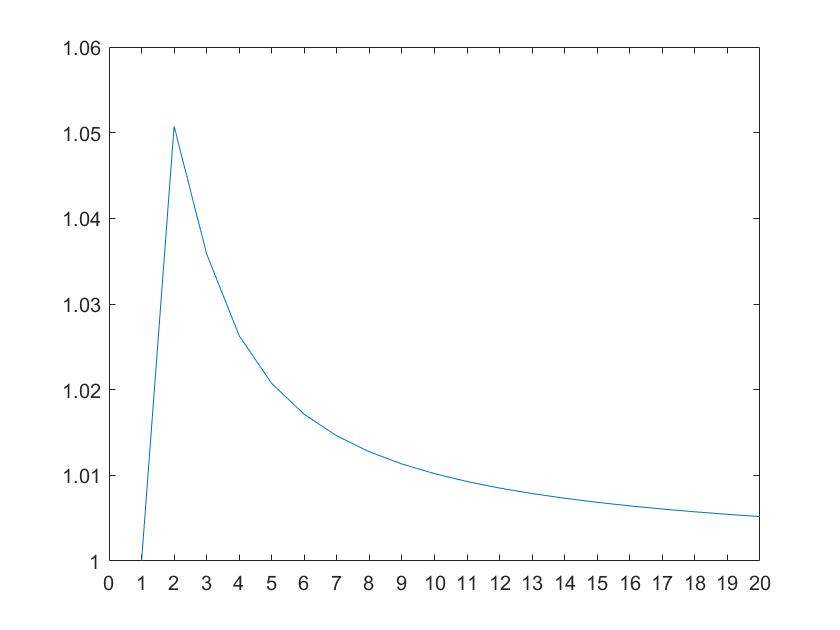}
         \caption{$\PA^{NR}_n(F)$.}
         \label{fig:PoA}
     \end{subfigure}
     \hfill
     \begin{subfigure}[b]{0.45\textwidth}
         \centering
         \includegraphics[width=\textwidth]{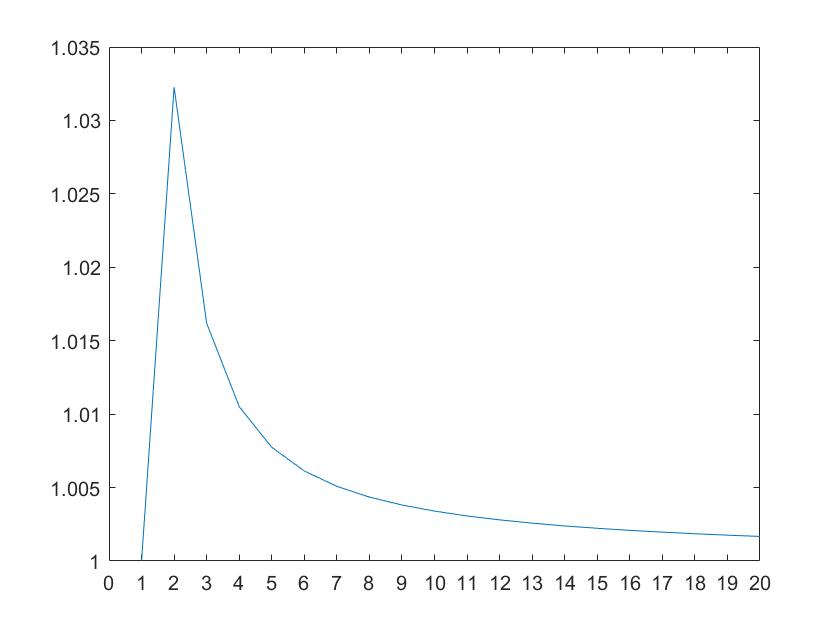}
         \caption{$\PS^{NR}_n(F)$}
         \label{fig:PoS}
     \end{subfigure}
     \hfill
     \begin{subfigure}[b]{0.45\textwidth}
         \centering
         \includegraphics[width=\textwidth]{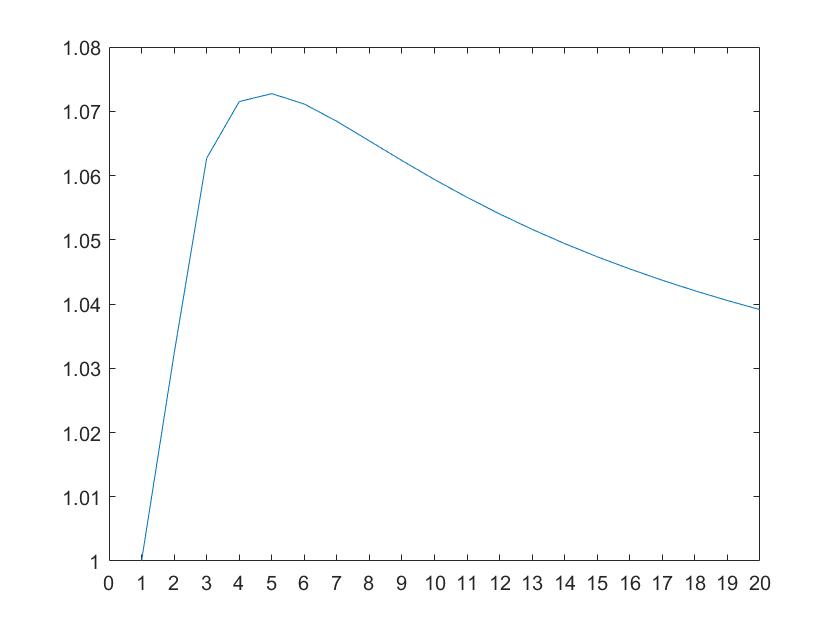}
         \caption{$\PR^{NR}_n(F)$}
         \label{fig:PR}
     \end{subfigure}
        \caption{Price of anarchy, price of stability and prophet ratio as function of $n$ when $F=\text{Unif}[0,1]$ in the no recall case.}
        \label{fig:ratios}
\end{figure}

\section{Proofs}\label{sec:proofs}

In this section we provide the proofs omitted in Section~\ref{sec:results}.
\subsection{Omitted proofs from Section~\ref{sec:payoffs}}
\paragraph{Full recall case.}
In this section, we prove Theorems~\ref{thm:SPEPcharact} and \ref{thm:caract}. Before that, we need some preliminary  results stating properties of the game $\Gamma^{FR}_n(a,b)$ defined in Section~\ref{sec:payoffs}. Recall that this game is defined as $\Gamma_n^{FR}$ with $n$ arrivals but with two initial values $a \geq b$ already present in the market. 

At first, in order to study the SPEP of $\Gamma_n^{FR}$, note that this game has the same SPEP as the game which is identical to  $\Gamma_n^{FR}$ and terminates at the first time a player gets an item or at stage $n$. 
If the game terminates because one player did get an item, the payoff of the other player is the value of the one-player continuation problem, and the payoff of the players if none of them did stop strictly before  stage $n$ is the expected payoff given by the pair of strategies "bid for the best available item". 
Indeed, it is easy to check that in both these situations the SPEP in the continuation games are unique and correspond to the payoffs of this auxiliary game. 
In the following, we assume that $\Gamma_n^{FR}$ is the auxiliary game. In this game, an history is just a sequence of values $(X_1,...,X_t)$  and the strategy of player $i$ at time $t<n$ is a measurable map $\sigma_{i,t}$ from histories into the probabilities over $\{\emptyset\} \cup \{1,...,t\}$. 
We use the same identification for the game $\Gamma^{FR}_n(a,b)$.   

Given an history $h=(x_1,...,x_t)$ of length $t$, let  $\Gamma_n(h)$ denote the subgame of $\Gamma^{FR}_{n+t}$ starting at stage $t$ after observing $h$ in which the two players are still present and $E_n(h)$ the set of SPEP of this game.  

In the following, variables $(X,X_1,X_2,....)$ will denote independent variables with distribution $F$, and $X_{(t)}=\max(X_1,...,X_t)$.

Without loss of generality, we assume that $\PP(X>0)>0$ (otherwise the set of equilibrium payoffs is reduced to $\{(0,0\}$).  

We state now without proofs some properties of SPEP which follow easily from usual arguments in dynamic game theory.
Recall that if $x \in [0,1] \rightarrow S(x)$ is a  set-valued map, $\EE[S(X)]$ is defined by
\[ \EE[S(X)]= \{ \EE[f(X)] \,|\, \text{$f$ measurable such that  $\forall x\in [0,1], f(x) \in S(x)$} \}. \]

\begin{lemma}\label{lemma:propertiesSPE}
The following properties hold:
\begin{enumerate}
\item $E_n=E_n(0,0)$
\item If $h=(x_1,...,x_t)$ is an history of length $t \geq 2$, then $E_n(h)= E_n(a,b)$ where $a$ and $b$ denote the first and second largest items in $h$ respectively.
\item $(x,y) \in E_n(a,b)$ with $a\geq b$ if and only if there exists $(d,e) \in \EE[ E_{n-1}(a\vee X,med[a,b,X])]$ such that $(x,y)$ is a mixed Nash equilibrium payoff of the finite game $G_n(a,b;d,e)$ with payoff matrix:
 \begin{table}[h]
 \centering
    \setlength{\extrarowheight}{7pt}
    \begin{tabular}{cc|c|c|c|}
      & \multicolumn{1}{c}{} & \multicolumn{2}{c}{Player $2$}\\
      & \multicolumn{1}{c}{} & \multicolumn{1}{c}{$a$}  & \multicolumn{1}{c}{$b$} & \multicolumn{1}{c}{$\emptyset$} \\\cline{3-5}
      \multirow{2}*{Player $1$}  & $a$ & $\left(\frac{a+ \EE(X_{(n)} \vee b)}{2} , \frac{a+ \EE(X_{(n)} \vee b)}{2} \right)$ & $(a,b)$ & $\left( a, \EE(X_{(n)} \vee b)\right)$ \\\cline{3-5}
      & $b$ & $(b,a)$ & $\left(\frac{b+ \EE(X_{(n)} \vee a)}{2} , \frac{b+ \EE(X_{(n)} \vee a)}{2} \right)$ & $(b, \EE(X_{(n)} \vee a))$ \\\cline{3-5}
       \multirow{2}*{}  & $\emptyset$ & $\left( \EE(X_{(n)} \vee b),a\right)$ &   $( \EE(X_{(n)} \vee a),b)$  &  $\left( d, e\right)$ \\\cline{3-5}
    \end{tabular}
    \caption{\label{tab:payoffs:general0} Payoff matrix of $G_n(a,b;d,e)$.}
 \end{table}
\item Similarly, $(\sigma_1,\sigma_2)$ is a pair of first-stage strategies of some SPE in $\Gamma_n^{FR}(a,b)$  with payoff $(x,y)$ in $\Gamma_n^{FR}(a,b)$ if and only if there exists $(d,e) \in \EE[ E_{n-1}(a\vee X,med[a,b,X])]$ such that $(\sigma_1,\sigma_2)$ is a mixed Nash equilibrium of the matrix game $G_n(a,b;d,e)$ with payoff $(x,y)$.
\end{enumerate}
\end{lemma}

The first point follows from the fact that any strategy in any subgame of $\Gamma_n(0,0)$ which bids with positive probability for one of the two initial items with value $0$ is strictly dominated by the modified strategy which waits until the last stage and bids for the best available item whenever the initial strategy bids for an item with value zero. Therefore such strategies are not played in any equilibrium in $\Gamma_n(0,0)$ and the result follows easily. 
The other points can be proven by induction on $n$ using the recursive structure of SPE, the one-shot deviation principle and measurable selection arguments.

In the following proposition, we prove that if a player prefers to pass instead of take $a$ given that the other player takes $a$, then he prefers to pass instead of take $a$ if the other player passes. 
\begin{proposition}\label{prop:ineq}
Let $a>b$, $(d,e) \in \EE[ E_{n-1}(a\vee X,med[a,b,X])]$ and $c=\EE[X_{(n)}\vee b]$.  If $c \geq a$ then $d \geq a$ and $e \geq a$. Similarly, if $c>a$ then $d>a$ and $e>a$.
\end{proposition}

By the symmetry of the game, it is enough to prove the inequality for $d$,  then the same arguments will hold for $e$. 

To prove the Proposition, we need the following two lemmas. 

\begin{lemma}\label{lem:tech:ineq}
If $X, Y$ are integrable random variables and $a$ a real number such that $\EE(X \vee Y)\geq a$ (resp. $>a$), then it holds that $\frac{1}{2}\EE(X \vee a)+ \frac{1}{2} \EE((X \wedge a) \vee Y)\geq a$ (resp. $>a$).
\end{lemma}
\begin{proof}

Let us first prove that if $x,y$ are real numbers, then
\begin{equation}\label{ineq:lemma}
x \vee y \leq x^++(y \vee (x \wedge 0)).
\end{equation}

Indeed, we show that for the four possible cases:
\begin{itemize}
    \item [(a)] If $x \vee y=x \geq 0$, the left hand side in \eqref{ineq:lemma} is $x$ and the right hand side is $x+\max\{y,0\}$ which is at least $x$ and the inequality holds. 
    \item [(b)] If $x \vee y=x < 0$, the the left hand side in \eqref{ineq:lemma} is $x$ and the right hand side is $0+\max\{y,x\}=x$ and \eqref{ineq:lemma} follows.  
    \item [(c)] If $x \vee y=y \geq 0$, the left hand side in \eqref{ineq:lemma} is $y$ and the right hand side is $x^++y$ which is at least $y$ and the inequality holds. 
    \item [(d)] If $x \vee y=y < 0$,   the left hand side in \eqref{ineq:lemma} is $y$ and the right hand side is $0+y$ and \eqref{ineq:lemma} follows.
\end{itemize}

Then, defining the random variable $Z=X \vee Y - X^+-(Y \vee (X \wedge 0)))$ and using the inequality above, \begin{equation}\label{ineq:rv}
    X \vee Y \leq X^++(Y \vee (X \wedge 0)) \ \ \ \ \text{a.s.}
\end{equation} 
Take $X'=X-a$ and $Y'=Y-a$, then by \eqref{ineq:rv} we have that 
\begin{equation}\label{ineq:rv_2}
    X' \vee Y' \leq X'^++(Y' \vee (X' \wedge 0)) \ \ \ \ \text{a.s.}
\end{equation}
But
\begin{itemize}
    \item [i)]$X' \vee Y' = (X-a) \vee (Y-a)= (X \vee Y)-a$,
    \item [ii)] $X'^+=(X-a)^+=(X \vee a)-a$, and
    \item [iii)] $Y' \vee (X' \wedge 0)= (Y-a) \vee (X-a \wedge 0)=Y\vee(X \wedge a)-a,$
\end{itemize} 
and thus \eqref{ineq:rv_2} means that 
\begin{equation*}
    (X \vee Y)-a \leq (X \vee a)-a +Y\vee(X \wedge a)-a \ \ \ \text{a.s.,}
\end{equation*}
which is equivalent to 
\begin{equation*}
    (X \vee Y)+a \leq (X \vee a) +Y\vee(X \wedge a) \ \ \ \text{a.s.}
\end{equation*}
In particular, the later implies that
\begin{equation*}
    \EE(X \vee Y)+a \leq \EE(X \vee a) +\EE(Y\vee(X \wedge a)),
\end{equation*}
and due to $ \EE(X \vee Y) \geq a$, we conclude that 
\begin{equation*}
    2a \leq \EE(X \vee a) +\EE(Y\vee(X \wedge a)),
\end{equation*}
and the result follows. The case with strict inequalities is similar.
\end{proof}

\begin{lemma}\label{lem:ineq:dn}
Under the same notation and assumptions as in Proposition \ref{prop:ineq}, we have 
\[ d \geq \frac{1}{2} \EE(X \vee a)+ \frac{1}{2}\EE(\beta(X)),\]
where  $X$ is a random variable with distribution $F$,  $\beta(x)= \EE[med[a,b,x]  \vee X_{(n-1)}]$ and $X_{(n-1)}=\max(X_1,..,X_{n-1})$ for some i.i.id sequence $X_1,...,X_{n-1}$ of samples of $X$.
\end{lemma}
\begin{proof}
By assumption, $d=\EE[f(X)]$ for some measurable function $f$ such that $f(x)$ is the expected payoff of player 1 in some SPE in $\Gamma_{n-1}^{FR}(a\vee x, med[a,b,x])$. We will obtain a lower bound for $f(x)$ by providing a lower bound for the payoff associated to a particular strategy in $\Gamma_{n-1}^{FR}(a\vee x, med[a,b,x])$. 

To this end, suppose that Player 1 plays the following strategy:
\begin{enumerate}
    \item  If $x \vee a \geq \beta(x)$, bid for $x \vee a$.
    \item If $x \vee a < \beta(x)$, wait until the end and get at least the second best item.
\end{enumerate}
Define $\Omega_1:= \{ x \in [0,1] \,:\, x \vee a \geq \beta(x) \}$, $\Omega_2:= \{ x \in [0,1] \,:\, x \vee a < \beta(x) \}$. Since the second best item in $\Gamma_{n-1}^{FR}(a\vee x, med[a,b,x])$ has expected value $\beta(x)$, we deduce that the payoff of player $1$, independently of the strategy of player $2$, is at least $\frac{1}{2}(x\vee a + \beta(x))$ if $x\in \Omega_1$ and  $\beta(x)$ if $x\in \Omega_2$. We deduce that 
\[ d = \EE[ f(X)] \geq  \EE \left[ \frac{1}{2}(X\vee a + \beta(X))\indi_{X\in \Omega_1}+\beta(X)\indi_{X\in \Omega_2} \right] .\]
Note that the last term is at least $\EE[ \frac{1}{2}( X \vee a+ \beta(X))\indi_{X\in \Omega_2}]$ because $x \vee a < \beta(x)$ for $x\in \Omega_2$.
Therefore, we conclude that 
\begin{eqnarray*}
    d &\geq& \EE \left[ \frac{1}{2}(X\vee a + \beta(X))\indi_{X\in \Omega_1}+\frac{1}{2}(X\vee a + \beta(X))\indi_{X\in \Omega_2} \right] \\
    &=& \EE \left[\frac{1}{2}(X\vee a + \beta(X))\right] \\
\end{eqnarray*}
 and the result follows. 
\end{proof}

\begin{proof}[Proof of Proposition \ref{prop:ineq}]
Take $X$ with distribution $F$ independent of some i.i.d sequence $(X_1,...,X_{n-1})$ of variables distributed according to $F$, and $Y=X_{(n-1)} \vee b$. Then $\EE(X \vee Y)= c \geq a$, and applying Lemma~\ref{lem:tech:ineq} we have
\[
\frac{1}{2}\EE(X \vee a)+ \frac{1}{2} \EE((X \wedge a) \vee X_{(n-1)} \vee b)\geq a.
\]
On the other hand, by independence and since $med[a,b,x]=(x\vee a)\wedge b$, we have 
\[ \EE((X \wedge a) \vee X_{(n-1)} \vee b) = \EE[\beta(X)].\]
By Lemma~\ref{lem:ineq:dn} $d \geq \frac{1}{2}\EE(X \vee a)+ \frac{1}{2} \EE(\beta(X)$. Putting all together we obtain $d \geq a$ and the proof is completed. 
The case with the strict inequality is similar.
\end{proof}

The next result allows to reduce the analysis of $G_n(a,b;d,e)$ to a smaller matrix game.

\begin{lemma}\label{lem:dominance}
Let $a\geq b$ in $[0,1]$, $n$ a positive natural number, and $(d,e) \in \EE[ E_{n-1}(a\vee X,med[a,b,X])]$. The game $G_n(a,b;e,d)$ has the same Nash equilibrium payoffs as the game with matrix:
   \begin{table}[h]
 \centering
    \setlength{\extrarowheight}{7pt}
    \begin{tabular}{cc|c|c|c|}
      & \multicolumn{1}{c}{} & \multicolumn{2}{c}{Player $2$}\\
      & \multicolumn{1}{c}{} & \multicolumn{1}{c}{$a$}  &  \multicolumn{1}{c}{$\emptyset$} \\\cline{3-4}
      \multirow{2}*{Player $1$}  & $a$ & $\left(\frac{a+ \EE(X_{(n)} \vee b)}{2} , \frac{a+ \EE(X_{(n)} \vee b)}{2} \right)$ & $\left( a, \EE(X_{(n)} \vee b)\right)$ \\\cline{3-4}
    
       \multirow{2}*{}  & $\emptyset$ & $\left( \EE(X_{(n)} \vee b),a\right)$ &  $\left( d, e \right)$ \\\cline{3-4}
    \end{tabular}
    
    \caption{\label{tab:payoff:matrix:reduced}Reduced payoffs matrix for $G_n(a,b;d,e)$.}
 \end{table}
\end{lemma}
\begin{proof}
We first consider the case $a>b$. 

For both players, the strategy $b$ is strictly dominated by mixed strategy consisting on playing $a$ with probability $1/2$ and pass with probability $1/2$.
Due to the symmetry of the payoffs' matrix, it is enough to show that for one player (let us say player 1) the expected payoff if he plays $a$ with probability $1/2$ and passes with probability $1/2$ is strictly higher that the expected payoff he obtains if he plays $b$.
Note that because  $a > b$ we have that:
\begin{itemize}
    \item[i-] $\frac{1}{2} \frac{a+ \EE(X_{(n)} \vee b)}{2} + \frac{1}{2} \EE(X_{(n)} \vee b) = \frac{a}{4} + \frac{3}{4}  \EE(X_{(n)} \vee b) > b, $
    \item[ii-] $\frac{1}{2} a + \frac{1}{2}  \EE(X_{(n)} \vee a) > \frac{b+ \EE(X_{(n)} \vee a)}{2},$ and
    \item[iii-] $ \frac{1}{2} a + \frac{1}{2} d > \frac{1}{2}b+ \frac{1}{2}b=b ,$
    \end{itemize}
    where the last inequality follows because $d \ge b$. Indeed, in any SPE, player $1$ obtains at least the second largest item (otherwise he could deviate to the strategy which does not bid until stage $n$), which is not smaller than $b$ in $\Gamma_n^{FR}(a,b)$.
We conclude that $b$ is a  strictly dominated strategy for both players (due to the symmetry of the game), and the result follows.

If $a=b$, we consider two subcases. If $\EE(X_{(n)} \vee a)>a$, then Proposition \ref{prop:ineq} implies that $d>a$ and $e>a$, so that the strategies $a$ and $b$ are strictly dominated by $\emptyset$. There is a unique Nash equilibrium$(\emptyset,\emptyset)$ with payoff $(d,e)$. Since the same holds for the matrix game given in table  \ref{tab:payoff:matrix:reduced}, the result follows.

If $\EE(X_{(n)} \vee a)=a$, the actions $a$ and $b$ are equivalent in the sense that they induce the same payoffs. Eliminating $b$ leads to a matrix game with the same Nash equilibrium payoffs. 
\end{proof}

Denoting $c=\EE(X_{(n)} \vee b)$, the payoff matrix introduced in Lemma \ref{lem:dominance} has the particular form exposed in Table~\ref{tab:payoff:cn}. Thus, it is enough to compute the NE of this matrix game where $(d,e)$ are parameters which correspond to the expected continuation equilibrium payoffs for players.

 \begin{table}[h]
 \centering
    \setlength{\extrarowheight}{4pt}
    \begin{tabular}{cc|c|c|c|}
      & \multicolumn{1}{c}{} & \multicolumn{2}{c}{Player $2$}\\
      & \multicolumn{1}{c}{} & \multicolumn{1}{c}{$a$}  & \multicolumn{1}{c}{$\emptyset$} \\\cline{3-4}
      \multirow{2}*{Player $1$}  & $a$ & $\left(\frac{1}{2}(a+c),\frac{1}{2}(a+c) \right)$ &  $\left( a, c\right)$ \\\cline{3-4}
    
       \multirow{2}*{}  & $\emptyset$ & $\left( c,a\right)$  &  $\left( d, e\right)$ \\\cline{3-4}
    \end{tabular}
    \caption{\label{tab:payoff:cn} General form of the reduced payoffs matrix for $G_n(a,b;d,e)$.}
  \end{table}

We are now ready to prove Theorem~\ref{thm:SPEPcharact}. 

\begin{proof}[Proof of Theorem~\ref{thm:SPEPcharact}]

Let us first analyze the game $G_n(a,b;d,d)$ with $(d,d)\in \EE[ E_{n-1}(a\vee X,med[a,b,X])]$, i.e. the case of continuation payoffs which belong to the diagonal in $\RR^2$, given in Table~\ref{tab:payoff:cn:final} below where $c=\EE(X_{(n)} \vee b)$. 

 \begin{table}[h]
 \centering
    \setlength{\extrarowheight}{4pt}
    \begin{tabular}{cc|c|c|c|}
      & \multicolumn{1}{c}{} & \multicolumn{2}{c}{Player $2$}\\
      & \multicolumn{1}{c}{} & \multicolumn{1}{c}{$a$}  & \multicolumn{1}{c}{$\emptyset$} \\\cline{3-4}
      \multirow{2}*{Player $1$}  & $a$ & $\left(\frac{1}{2}(a+c),\frac{1}{2}(a+c) \right)$ &  $\left( a, c\right)$ \\\cline{3-4}
    
       \multirow{2}*{}  & $\emptyset$ & $\left( c,a\right)$  &  $\left( d, d\right)$ \\\cline{3-4}
    \end{tabular}
    \caption{\label{tab:payoff:cn:final} Payoffs matrix for $G_n(a,b;d,d)$.}
  \end{table}

Let us compute the mixed Nash equilibria of this game. Using Proposition \ref{prop:ineq}, $c>a \Rightarrow d>a$  and $c\geq a \Rightarrow d\geq a$. We have that:
\begin{itemize}
\item[(a)] If $a> c \vee d$, then $(a,a)$ is the unique NE with payoff $\left(\frac{1}{2}(a+c),\frac{1}{2}(a+c)\right)$.
\item[(b)] If $d>a>c$, there are two pure NE $(a,a)$ and $(\emptyset, \emptyset)$ and a symmetric mixed equilibrium in which both agents play $a$ with probability $\frac{2(d-a)}{2d-a-c} $ and pass with probability $\frac{a-c}{2d-a-c}.$ Furthermore, the equilibrium payoffs are: 
\[\left(\frac{1}{2}(a+c), \frac{1}{2}(a+c)\right) , (d,d)  \text{ and } \left( \frac{dc-2ac+ad}{2d-a-c}, \frac{dc-2ac+ad}{2d-a-c}\right), \text{ respectively.}\]
\item[(c)] If $a <c$, then $(\emptyset, \emptyset)$ is the unique NE, and $(d,d)$ is the expected payoff. 
\item[(d)] If $d=a>c$ or $d>a=c$, there are two pure NE $(a,a)$ and $(\emptyset, \emptyset)$ with payoffs: 
\[\left(\frac{1}{2}(a+c), \frac{1}{2}(a+c)\right)  \text{ and } (d,d) .\]
\item[(e)] If $d=a=c$, then any profile is a NE with payoff $(d,d)$.
\end{itemize}
  
From the analysis above we conclude that the set of NE payoffs of the game is contained in the diagonal. 

We now prove by induction on $n$ that $E_n(a,b)$ is also a subset of the diagonal, that is:
\[   E_n(a,b) \subset \{ (u,u) : u \in \RR_+\}. \]
Ar first, $E_0(a,b)=\{ ( \frac{a+b}{2}, \frac{a+b}{2})\}$, so the statement is correct for $n=0$.

Let us assume that the statement is correct for $E_n(a,b)$ for all pairs $(a,b)$.
Then, it implies that $\EE[ E_n( a\vee X, med[a,b,X])]$ is also a subset of the diagonal. Therefore, by point 3) of Lemma \ref{lemma:propertiesSPE} and the above analysis, we deduce that $E_{n+1}(a,b)$ is a subset of the diagonal. 

The first statement of the theorem is proved. 
  
Let us now show that $\min \text{P}E_n(a,b)=l_n(a,b)$ and  $\max \text{P}E_n(a,b)=h_n(a,b)$, where $\text{P}E_n(a,b)$ is the projection of $E_n(a,b)$ to its first coordinate in $\RR$.
  
Note that from the analysis done above, we have that if $a>c \vee d$ or $a < c$ or $a=c=d$, there is only one NE payoff but in the other cases we have multiple equilibrium payoffs. When playing a SPE, the expected payoff of the players when both pass depends on which equilibrium is played in the following stages of the game. It is known that, in general, we cannot assume that if ``the worst'' or ``the best'' equilibrium is played at each stage, that results in the worst or best equilibrium of the game. However, below we show that this is indeed true for the game we are considering and then the result will follow just computing the expected payoff corresponding to the best and worst equilibrium. 

To this end, it is enough to prove that the expected payoff of a player is increasing in $d$. 
  
Define $\Omega=\{(a,c,d) \in \RR^3 : ( c>a \Rightarrow d>a) \text{ and } ( c \geq a \Rightarrow d \geq a)   \} $ and consider the multivalued function $\psi: \Omega \rightrightarrows \RR$ defined by
    
  \[   
\psi(a,c,d) = 
     \begin{cases}
       d & \text{if } c>a \text{ or } a=c=d\\
      \left\{ \frac{a+c}{2},d\right\} &\text{if } c<a<d \text{ or }  c=a<d\text{ or } c<a=d \\
      \frac{a+c}{2} & \text{if } c\vee d < a.
      
     \end{cases}
\]
Note that if $d>a>c$, then  
\[ 
d>\frac{dc-2ac+ad}{2d-a-c}> \frac{1}{2} (a+c),
 \]
and thus  $d$ and $\frac{1}{2} (a+c)$ are respectively the "best" and the "worst" equilibrium payoffs for player $1$ in $G_n(a,b;d,e)$. Therefore, $\psi(a,c,d)$ represents the "best" and "worst"  NE payoffs for player $1$ in the game represented by Table~\ref{tab:payoff:cn:final}.

We say that the multivalued function is non-decreasing in $d$ if for each $a,c,d_1,d_2$ such that $d_1 < d_2$ and $(a,c,d_i) \in \Omega$ for $i=1,2$, holds that $\min \psi (a,c,d_1) \leq \min \psi(a,c,d_2)$ and $\max \psi(a,c,d_1) \leq \max \psi(a,c,d_2)$. Let us see that $\psi$ is non-decreasing in $d$. To show that, fix $a$, $c$ and take $d_1,d_2$ such that $d_1 < d_2$ and $(a,c,d_1) \in \Omega$, $(a,c,d_2) \in \Omega$. We have the following cases:
\begin{itemize}
    \item[(a)] If $c>a$, then $\psi(a,c,d_i)=\{d_i\}$ for $i=1,2$ and the result is obvious.  
    \item[(b)] If $c<a \leq d_1$, then $c<a<d_2$ and $\frac{a+c}{2} < d_1<d_2$. Thus, $\min \psi(a,c,d_1)=\frac{a+c}{2}= \min \psi(a,c,d_2)$ and $\max \psi(a,c,d_1)=d_1 \leq d_2= \max \psi(a,c,d_2).$
    \item[(c)] If $a>c \vee d_2$, then $\psi(a,c,d_1)=\{\frac{a+c}{2}\}=\psi(a,c,d_2)$ and the result is obvious.
    \item[(d)] If  $c \vee d_1 < a \leq  d_2$, then $\psi(a,c,d_1)=\{\frac{a+c}{2}\}$ and  $\psi(a,c,d_2)=\left\{ \frac{a+c}{2},d_2\right\}$, obtaining $\min\psi(a,c,d_1)=\max\psi(a,c,d_1) \leq \min \psi(a,c,d_2) \leq \max \psi(a,c,d_2)$.
    \item[(e)] If $a=c \leq d_1$, then $\psi(a,c,d_1)=\{a\}$ or $\{a,d_1\}$ and $\psi(a,c,d_2)=\{a,d_2\}$, and the result follows.
\end{itemize}
Thus, we conclude that $\psi$ is non-decreasing in $d$.
The monotonicity property we proved for $\psi$ means that playing a ``better'' equilibrium in the continuation game gives a ``better'' equilibrium for $\Gamma_n^{FR}(a,b)$ and playing a ``worse'' equilibrium one gives a ``worse'' equilibrium for $\Gamma_n^{FR}(a,b)$.

Precisely, define $L: \Omega \rightarrow \RR $ and $H:\Omega \rightarrow \RR$ by:
\[   
L(a,c,d) = 
     \begin{cases}
       d & \text{ if } a < c ,\\
      \frac{1}{2}(a+c) & \text{ if } a \geq c\\ 
     \end{cases}
 \; \; {\it and}\;\;    
H(a,c,d) = 
     \begin{cases}
      \frac{1}{2}(a+c) & \text{if } a > c \vee d,\\
      d  &\text{if }  a \leq c \vee d\\ 
     \end{cases}
\]
$L(x,y,z)$ and $H(x,y,z)$ represent the lowest and highest equilibrium expected payoffs of player $1$ in $G_n(a,b;d,d)$. 
As we aforementioned, we are interested in computing the expected payoff corresponding to the best and the worst NE of the game $\Gamma_n^{FR}(a,b)$, which correspond to the extremes values of the set $\E_n(a,b)$ denoted $l_n(a,b)$ and $h_n(a,b)$. 
Using the monotony of $\psi$ (and thus of $L$ and $H$) with respect to $d$, $l_n(a,b)$ is the lowest equilibrium payoff of $G_n(a,b;d,d)$, i.e. $L(a,\EE[X_{(n)}\vee b],d)$, when $d$ is the expected continuation payoff obtained by playing the lowest equilibrium in every continuation game, that is:  
\[ d=\min \EE[ E_{n-1}(a\vee X,med[a,b,X])] = \EE_X\left[ l_{n-1}(a \vee X,\text{med}[a,b,X])\right].\] 
where $l_0(a,b)=\frac{a+b}{2}$. This completes the proof for $l_n$ and a the same analysis applies for $h_n(a,b)$.
\end{proof}

\begin{remark}
Note that we did not use any assumption on the distribution $F$ except that it is supported by $[0,1]$, therefore Theorem~\ref{thm:SPEPcharact} holds when considering discrete distributions.
\end{remark}
Using Theorem~\ref{thm:SPEPcharact}, we now prove Theorem~\ref{thm:caract}.

 \begin{proof}[Proof of Theorem~\ref{thm:caract}]
 By Theorem~\ref{thm:SPEPcharact}, we know that $E^{FR}_n \subset \{ (u,u) : l_n \leq u \leq h_n\}$.
 
Furthermore, 
 $E^{FR}_n= \left\{ \int_{[0,1]} f(x) \mathrm{d}F(x) : f(x) \in \E_{n-1}(x,0) \right\},$ and then $(l_n,l_n)$ and $(h_n,h_n)$ belongs to $E^{FR}_n$ (we obtain them just taking $f(x)=(l_{n-1}(x,0),l_{n-1}(x,0))$ and $f(x)=(h_{n-1}(x,0),h_{n-1}(x,0))$, respectively).
To obtain the result is then enough to prove that the set $E_n$ is convex. To this end, let us define the function $\mu: [0, 1] \rightarrow \RR$ by 
\[
\mu(\alpha)=\int_{[0\alpha)} l_{n-1}(x,0) \mathrm{d}F(x)+\int_{[\alpha,1]} h_{n-1}(x,0) \mathrm{d}F(x).
\]
Notice that $\mu$ is continuous since $F$ is atomless, $\mu(0)=h_{n}$ and $\mu(1)=l_{n}$, then using the intermediate value theorem all values between $l_n$ and $h_n$ are taken by $\mu$. But $(\mu(\alpha),\mu(\alpha))$ belongs to $E_n$ for every $\alpha$ by choosing $f(x)$ equal to $(l_{n-1}(x,0),l_{n-1}(x,0))$ on $[0,\alpha$ and to $(h_{n-1}(x,0),h_{n-1}(x,0))$ on $[\alpha,1]$. Therefore $E_n$ is convex.
 \end{proof}

\paragraph{No recall case.} 
We pass now to the no recall case and the goal is to prove Theorem~\ref{thm:norecall}. 

Recall that the game $\Gamma_n^{NR}(a)$ is defined as $\Gamma_n^{NR}$, but there are $n+1$ stages: at the first stage, the players can bid for an item with value $a$, and the items for the next $n$ stages are randomly drawn as in $\Gamma_n^{NR}$. 

At first, in order to study the SPEP of $\Gamma_n^{NR}$, note that this game has the same SPEP as the game which is identical to  $\Gamma_n^{NR}$ but terminates at the first time a player gets an item (or after stage $n$). 
If the game terminates because one player did get an item, the payoff of the other player is the value of the one-player continuation problem. As for the full recall case, it is easy to check that if after a player gets an item, the SPEP in the continuation game are unique and correspond to the payoffs of this auxiliary game. 
In the following, we assume that $\Gamma_n^{NR}$ is the auxiliary game. In this game, an history is just a sequence of values $(X_1,...,X_t)$  and the strategy of player $i$ at time $t<n$ is a measurable map $\sigma_{i,t}$ from histories into the probabilities over $\{\emptyset\} \cup \{1,...,t\}$. 
We use the same identification for the game $\Gamma^{NR}_n(a)$.   

Given an history $h=(x_1,...,x_t)$ of length $t$, let  $\Gamma^{NR}_n(h)$ denote the subgame of $\Gamma^{NR}_{n+t}$ starting at stage $t$ after observing $h$ in which the two players are still present and $E_n^{NR}(h)$ the set of SPEP of this game.  

In the following, variables $(X,X_1,X_2,....)$ will denote independent variables with distribution $F$, and $X_{(t)}=\max(X_1,...,X_t)$.

Without loss of generality, we assume that $\PP(X>0)>0$ (otherwise the set of equilibrium payoffs is reduced to $\{(0,0\}$).  

As for the full recall case, we state without proofs some properties of SPEP which follow easily from usual arguments in dynamic game theory.

\begin{lemma}\label{lemma:propertiesSPE2}
The following properties hold:
\begin{enumerate}
\item $E_0^{NR}=\{(0,0)\}$ and for $n\geq 1$, $E_n^{NR}=E^{NR}_n(0)=\EE[E_{n-1}^{NR}(X)]$.
\item If $h=(x_1,...,x_t)$ is an history of length $t \geq 1$, then $E_n^{NR}(h)= E_n^{NR}(a)$ where $a$ denotes the largest item in $h$.
\item $(x,y) \in E^{NR}_n(a)$ if and only if there exists $(d,e) \in E^{NR}_{n}$ such that $(x,y)$ is a mixed Nash equilibrium payoff of the finite game $G_n(a;d,e)$ with payoff matrix:
 \begin{table}[h]
 \centering
    \setlength{\extrarowheight}{4pt}
    \begin{tabular}{cc|c|c|c|}
      & \multicolumn{1}{c}{} & \multicolumn{2}{c}{Player $2$}\\
      & \multicolumn{1}{c}{} & \multicolumn{1}{c}{$a$}  & \multicolumn{1}{c}{$\emptyset$} \\\cline{3-4}
      \multirow{2}*{Player $1$}  & $a$ & $\left(\frac{1}{2}(a+c_n),\frac{1}{2}(a+c_n) \right)$ &  $\left( a, c_n\right)$ \\\cline{3-4}
    
       \multirow{2}*{}  & $\emptyset$ & $\left( c_n,a\right)$  &  $\left( d, e\right)$ \\\cline{3-4}
    \end{tabular}
    \caption{\label{tab:payoff_no_recall0};Payoffs matrix of $G_n(a;d,e)$.}
  \end{table}
where $c_n$ denotes the value of the decision problem with $n$ stages in a standard prophet setting.
\item Similarly, $(\sigma_1,\sigma_2)$ is a pair of first-stage strategies of some SPE in $\Gamma_n^{NR}(a)$  with payoff $(x,y)$ if and only if there exists $(d,e) \in E^{NR}_{n}$ such that $(\sigma_1,\sigma_2)$ is a mixed Nash equilibrium of the matrix game $G_n(a;d,e)$ with payoff $(x,y)$.
\end{enumerate}
\end{lemma}

 Before concluding the section with the proof of Theorem~\ref{thm:norecall}, we show Proposition~\ref{prop:convexity}.


\begin{proof}[Proof of Proposition~\ref{prop:convexity}.]
The fact that $E^{NR}_n$ is is symmetric with respect to the diagonal is a direct consequence of the fact the the game is symmetric. 

To prove convexity, we can use the properties of the expectation of a set-valued map, also called the Aumann integral. Theorem 8.6.3 in \cite{AF09} 
implies that the expectation of a set-valued map with non-empty closed values and compact graph from $[0,1]$ to $\RR^2$ with respect to an atomless measure on $[0,1]$ is a non-empty convex compact set.

At first, $E^{NR}_0=\{(0,0)\}$ is non-empty compact convex. Then, $a \rightarrow E^{NR}_0(a)$ is a set-valued map with non-empty closed values and compact graph using the classical properties of Nash equilibrium payoffs of matrix games. It follows that $E^{NR}_1$ is non-empty compact convex. Let us assume that $E^{NR}_n$ is non-empty compact convex. Let $NEPG_n(a;d,e)$ denote the set of Nash equilibrium payoffs of $G_n(a;d,e)$, then $(a,d,e) \in [0,1]\times E^{NR}_n \rightarrow NEPG_n(a;d,e)$ is a set-valued map with non-empty closed values and a compact graph, and thus $a \rightarrow E^{NR}_n(a)$ is a set-valued map with non-empty closed values and a compact graph. We conclude that $E^{NR}_{n+1}$ is a  non-empty compact convex set.    
\end{proof}

Below, we prove the theorem.

\begin{proof}[Proof of Theorem~\ref{thm:norecall}.]
Given $a \in [0,1]$ and $n$ a natural number, we consider the game $\Gamma_n^{NR}(a)$ defined in Section~\ref{sec:payoffs}. We first analyze the game $G_n(a;d,e)$ described in Table~\ref{tab:payoff_no_recall0} with $(d,e)\in E^{NR}_n$.

One first remark to do regarding this game, is that what a player gets if he stays alone in the game, that is $c_n$, is at least what he gets if both stay. In other words, $c_n\ge d $ and $c_n \ge e$.
  

Now, we use Table~\ref{tab:payoff_no_recall0} and the remark above to study the NE of the game $G_n(a;d,e)$, depending on the relation between the parameters $a,c_n,d,e$. It is easy to check that the following holds:
\begin{itemize}
     \item[(a)] If $a>c_n$, $(a,a)$ is the unique NE with payoff $\left( \frac{1}{2}(a+c_n),\frac{1}{2}(a+c_n)\right)$.
     \item[(b)] If $a=c_n$, there is a unique NE payoff $(c_n,c_n)$.
     \item[(c)] If $c_n>a>d \vee e$ there are two pure NE $(a,\emptyset)$ and $(\emptyset,a)$, and a mixed equilibrium; with payoffs $\left( a,c_n\right), (c_n,a)$ and $(\gamma_n^1,\gamma_n^2)$, respectively, with $\gamma_n^1=\frac{2ac_n-d(c_n+a)}{c_n+a-2d}, \gamma_n^2=\frac{2ac_n-e(c_n+a)}{c_n+a-2e}$. 
     \item[(d)] If $c_n>a=d>e$, the NE payoffs are $(c_n,a)$ and $(a,\lambda)$ with $\lambda \in [\gamma_n^2,c_n]$.
     \item[(e)] If $c_n>a=e>d$, the NE payoffs are $(a,c_n)$ and $(\lambda,a)$ with $\lambda \in [\gamma_n^1,c_n]$.
     \item[(f)] If $c_n> a=d=e$, the NE payoffs are $(a,\lambda)$ and $(\lambda,a)$ with $\lambda \in [a,c_n]$.
      \item[(g)] If $d \wedge e>a$, $(\emptyset,\emptyset)$ is the unique NE with payoff $(d,e)$.
      \item[(h)] If $d>a>e$, $(\emptyset,a)$ is the unique NE with payoff $(c_n,a)$.
     \item[(i)] If $e>a>d$, $(a, \emptyset)$ is the unique NE with payoff $(a,c_n)$.
     \item[(j)] If $d> a=e$, the NE payoffs are $(d,e)$ and $(\lambda,a)$ with $ \lambda \in [d,c_n]$.
     \item[(k)] If $e>a=d$, the Ne payoffs are $(d,e)$ and $(a,\lambda)$ with $\lambda \in [e,c_n]$.
 \end{itemize}


In particular, we are interested on computing the sum of the expected payoff corresponding to the best and the worst SPEs that is, $\max\{x+y: (x,y) \in E^{NR}_{n+1}\}$ and  $\min\{x+y: (x,y) \in E^{NR}_{n+1}\}$, respectively; as well as the worst payoff a player can get at equilibrium, that is $ \min \{ \min \{ x,y \} : (x,y) \in E^{NR}_n\}$.

 Due to Proposition~\ref{prop:convexity} we have that
 \[
 \max\{x+y: (x,y) \in E^{NR}_{n+1}\}=2 \max\{x: (x,x) \in  E^{NR}_{n+1}\},
 \]
 and
 \[
 \min\{x+y: (x,y) \in E^{NR}_{n+1}\}=2 \min\{x: (x,x) \in  E^{NR}_{n+1}\}.
 \]
 Therefore, defining $\alpha_n:= \min \{ x : (x,x) \in E^{NR}_n\}$, $\beta_n:= \max \{x : (x,x) \in E^{NR}_n \}$ and $\alpha'_n= \min \{ \min\{x,y\} : (x,y) \in E^{NR}_n\}$, it follows that it is enough to compute $2 \alpha_{n+1}$, $2\beta_{n+1}$ and $\alpha'_{n+1}$.
 
 To prove the part $a)$ of the theorem, we compute a recursive formula for $\alpha_{n+1}'$. Define $\alpha'_{n}(a)$ as the minimal NE payoff of player $1$ in the family of game $G_n(a;d,e)$ when $(d,e)$ ranges through $E^{NR}_n(a)$. It is clear from the previous characterization that $\alpha'_{n+1}=\int_{[0,1]} \alpha'_{n}(a)\mathrm{d}F(a)$.
Using our previous analysis, notice that:
\[   
\alpha_n'(a) = 
     \begin{cases}
       \alpha'_n & \text{ if } a < \alpha'_n ,\\
      a & \text{ if } \alpha'_n <a< c_n,\\ 
      \frac{a+c_n}{2} & \text{ if } a>c_n.
     \end{cases}
\]
Since $F$ is atomless, it is sufficient to conclude that for $n\ge 1$
\begin{eqnarray}\label{eq:alpha'} \nonumber
\alpha'_{n+1}&=&\int_{0}^{\alpha'_n} \alpha'_n \mathrm{d}F(a)+\int_{\alpha'_n}^{c_n}a \mathrm{d}F(a) + \int_{c_n}^{1} \frac{a+c_n}{2} \mathrm{d}F(a) \\ \nonumber
&=&{\alpha'_n} (F({\alpha'_n})-F(0))+(c_n F(c_n)-\alpha'_n  F(\alpha'_n))-\int_{\alpha'_n}^{c_n} F(a) \mathrm{d}a \\ \nonumber
&+& \frac{c_n}{2}(F(1)-F(c_n))+\frac{1}{2}(1  F(1)-c_n F(c_n))-\frac{1}{2}\int_{c_n}^{1} F(a) \mathrm{d}a. \\
&=&\frac{(c_n+1)}{2}-\int_{\alpha'_n}^{c_n} F(a) \mathrm{d}a -\frac{1}{2}\int_{c_n}^{1} F(a) \mathrm{d}a,
\end{eqnarray}
and the first statement of Theorem~\ref{thm:norecall} is proved. 

Next,  we compute $2\beta_{n+1}$ as a function of $\beta_n, \alpha'_n$ and $c_n$.
Define $2\beta_n(a)$ as the maximal sum of payoffs in any NE on the family of game $G_n(a;d,e)$ when $(d,e)$ ranges through $E^{NR}_n(a)$, so that $\beta_{n+1}=\int_{[0,1]} \beta_{n}(a)\mathrm{d}F(a)$.  We have:
\[   
2\beta(a) = 
     \begin{cases}
       2\beta_n & \text{ if } a < \alpha'_n ,\\
      \max\{a+c_n,2\beta_n\} & \text{ if } \alpha'_n <a< \beta_n,\\ 
      a+c_n & \text{ if } a>\beta_n.
     \end{cases}
\]
Since $F$ is atomless, it is sufficient to conclude that for $n \ge 1$
 
 \begin{eqnarray}\label{eq:beta}
 2\beta_{n+1}&=& \int_{0}^{\alpha'_n} 2 \beta_n \mathrm{d}F(a) + \int_{\alpha'_n}^{\beta_n}  \max\{a+c_n, 2\beta_n \} \mathrm{d}F(a) +\int_{\beta_n}^{1} (a+c_n) \mathrm{d}F(a),
 \end{eqnarray}
and the second statement of the theorem is obtained. 

It remains to compute $2 \alpha_{n+1}.$
As before, define $2\alpha_n(a)$ as the minimal sum of payoffs in any NE on the family of game $G_n(a;d,e)$ when $(d,e)$ ranges through $E^{NR}_n(a)$, so that $\alpha_{n+1}=\int_{[0,1]} \alpha_{n}(a)\mathrm{d}F(a)$. Notice that if  $c_n>a>d \vee e$, the mixed equilibrium gives a worse sum of payoffs than the pure.
We have:
\[   
2\alpha_n(a) = 
     \begin{cases}
       2\alpha_n & \text{ if } a < \alpha'_n ,\\
   \min\{2\alpha_n, a+c_n\} & \text{ if }  \alpha'_n<a<\alpha_n ,\\
      2a & \text{ if } \alpha_n <a< \beta_n,\\ 
      2 \frac{2ac_n-\beta_n(a+c_n)}{c_n+a-2 \beta_n} & \text{ if } \beta_n <a< c_n,\\ 
      \frac{a+c_n}{2} & \text{ if } a>c_n.
     \end{cases}
\]
Since $F$ is atomless, it is sufficient to conclude that for $n \ge 1$
\begin{eqnarray}\label{eqn:alphan}
 \alpha_{n+1}= \int_{c_{n}}^{1} (a+c_{n})\mathrm{d}F(a) + \int_{\beta_{n}}^{c_{n}} \psi_n(a) \mathrm{d}F(a) + \int_{\alpha_{n}}^{\beta_{n}} 2 a \mathrm{d}F(a) 
+ \int_{\alpha_{n}'}^{\alpha_{n}} \xi_n(a) \mathrm{d}F(a) +  \int_{0}^{\alpha_{n}'}  2 \alpha_{n} \mathrm{d}F(a),
 \end{eqnarray}
where $\psi_n(a)=\frac{4ac_{n}-2\beta_{n}(a+c_{n})}{c_{n}+a-2 \beta_{n}} $ and $\xi_n(a)=\min \{2 \alpha_{n}, a+c_{n}\}$, which is the third statement of the theorem. 

Putting together all the foregoing analysis, we obtain the desired result, concluding the proof. 
\end{proof}

 
\subsection{Omitted proofs from Section~\ref{sec:comparison}} \label{app:comparison}

\begin{proof}[Proof of Lemma \ref{lem:bound:SPEP}.]
 We assume that one player, let us say player 1, bid in the first stage if $a\vee X_1 \geq c_n$, and passes otherwise, where $X_1$ is the realization of the first random variable arrived. Let us show that player 1 obtain an expected payoff, namely $\gamma$, of at least  $(a+c_{n+1})/2$, independently of what player 2 does, and therefore player 1 will obtain a payoff of at least $(a+c_{n+1})/2$ playing any SPE.  
 
 We divide the proof in two cases:
 \paragraph{Case 1.} Assume that $a\geq c_n.$ Note that in this case, player 1 bids for $M=a\vee X_1,$ and we have that 
 \begin{equation*}
     \gamma = \EE_{X_1}\left( (a \vee X_1) \PP(A)+\frac{1}{2}\left((a \vee X_1)+ \EE\left((a \wedge X_1)\vee X_{(n)}\right)\right) \PP(A^c) \right),
 \end{equation*}
where $A$ is the event \textit{player 2 does not bid for M}. 

Notice that  $\left((a \vee X_1)+ \EE\left((a \wedge X_1)\vee X_{(n)}\right)\right)$ is lower bounded by $a+c_n$ and by $a+X_1$ and therefore we have that 
 \begin{eqnarray*}
     \gamma &\geq& \EE_{X_1}\left( (a \vee X_1) \PP(A)+\frac{1}{2}\left((a +c_n) \textbf{1}_{\{X_1<c_n\}}+ (a + X_1)  \textbf{1}_{\{X_1 \geq c_n\}} \right) \PP(A^c) \right) \\
     &=& \frac{a}{2}+ \frac{1}{2}\left( \PP(A) (2(a \vee X_1)-a)+ \PP(A^c) \left( c_n \textbf{1}_{\{X_1<c_n\}}+ X_1 \textbf{1}_{\{X_1 \geq c_n\}} \right)   \right) \\
     &\geq & \frac{a}{2}+\frac{1}{2} \left( c_n \textbf{1}_{\{X_1<c_n\}}+ X_1 \textbf{1}_{\{X_1 \geq c_n\}} \right)  =  \frac{a}{2}+ \frac{c_{n+1}}{2},
 \end{eqnarray*}
 where the second inequality holds because $2(a \vee X_1)-a \geq  c_n \textbf{1}_{\{X_1<c_n\}}+ X_1 \textbf{1}_{\{X_1 \geq c_n\}}$ due to $a>c_n$ and the last equality because $c_{n+1}=\EE(X\vee c_n).$ We conclude that $\gamma\geq (a+c_{n+1})/2$ if $a \geq c_n.$
 \paragraph{Case 2.} Assume that $a< c_n.$ In this case, we have that 
 \begin{eqnarray*}
 \gamma &\geq& \EE_{X_1}\left( \left( X_1 \wedge \frac{1}{2} \left( X_1 + \EE\left(a \vee X_{(n)}\right)\right) \right) \textbf{1}_{\{X_1 \geq c_n\}} + \left(\frac{(a \vee X_1)+c_n}{2} \wedge \EE\left( X_{(n)} \right) \right) \textbf{1}_{\{X_1<  c_n\}}\right) \\
 &=& \EE_{X_1}\left(  \left( \frac{X_1}{2} + \frac{1}{2} \left( X_1 \wedge \EE\left(a \vee X_{(n)}\right)\right) \right) \textbf{1}_{\{X_1 \geq c_n\}}+ \left( \frac{c_n}{2} + \frac{1}{2} \left( \left(a\vee X_1\right) \wedge \left( 2 \EE\left(X_{(n)}\right) -c_n \right)\right) \right) \textbf{1}_{\{X_1<  c_n\}} \right) \\
 &=& \frac{c_{n+1}}{2}+ \EE_{X_1}\left(  \frac{1}{2} \left( X_1 \wedge \EE\left(a \vee X_{(n)}\right)\right) \textbf{1}_{\{X_1 \geq c_n\}} + \frac{1}{2} \left( \left(a\vee X_1\right) \wedge \left( 2 \EE\left(X_{(n)}\right) -c_n \right)\right) \textbf{1}_{\{X_1<  c_n\}}  \right) \\
 &\geq& \frac{a+c_{n+1}}{2},
 \end{eqnarray*}
 where the second equality holds because $c_{n+1}=\EE(X\vee c_n)$ and the last inequality because both $ X_1 \wedge \EE\left(a \vee X_{(n)}\right)$ and $\left(a\vee X_1\right) \wedge \left( 2 \EE\left(X_{(n)}\right) -c_n \right)$ are at least $a$ when $a<c_n$. Thus, we conclude that $\gamma \geq (a+c_{n+1})/2$ if $a<c_n$.
 
 Putting all together we obtain that $\gamma_n^{FR}\geq (a+c_{n+1})/2$ and the proof is complete. 
 \end{proof}
 \begin{proof} [Proof of Theorem \ref{thm:comparison}]
 Note that it is enough to prove that $l_n \geq \beta_n$ for all $n$, where $l_n$ denotes the lowest SPE payoff in the full recall case and $\beta_n$ denotes the highest symmetric SPE payoff in the no recall case. To this end, we will prove by induction on $n$ that $l_n(a,0) \geq \beta_n(a)$ for all $n$ and for all $a$, where $l_n(a,0)$ is defined as in the statement of Theorem~\ref{thm:SPEPcharact} and $\beta_n(a)$ is defined by 
   \[   
\beta_n(a) = 
     \begin{cases}
     \beta_n & \text{if } a < \alpha'_n\\
       \max\{\beta_n, \frac{a+c_n}{2}\}  &\text{if }  a \in [\alpha'_n, \beta_n] \\
        \frac{a+c_n}{2}  &\text{if } a > \beta_n, \\
     \end{cases}
\]
 where $c_n$ is the value of the decision problem in the no recall case with one decision-maker and $n$ arrivals. 
 
 First, notice that $l_1(a,0)=(a+\EE(X))/2$, $c_1=\EE(X)$ and $\beta_1=\EE(X)$, and thus $l_1(a,0) \geq \beta_1$ and $l_1(a,0) \geq (a+c_1)/2$ for all $a$, concluding that $l_1(a,0) \geq \beta_1(a)$ for all $a$.
 
 We assume now that $l_n(a,0) \geq \beta_n(a)$ for all $a$ and we prove that $l_{n+1}(a,0) \geq \beta_{n+1}(a)$ for all $a$.
 We divide the proof in two cases depending on if $a$ is higher than $\EE\left(X_{(n+1)}\right)$ or not. 
 \paragraph{Case 1.} Assume that $a>\EE\left(X_{(n+1)}\right)$. In this case, we have that $a>\beta_{n+1}$ because $\EE\left(X_{(n+1)}\right)\geq c_{n+1}\geq \beta_{n+1}.$ Thus, by the definition of $\beta_{n+1}(a),$ we have that 
 \begin{equation}\label{eqn:comp1}
     \beta_{n+1}(a)=\frac{a+c_{n+1}}{2}.
 \end{equation}

 On the other hand, holds that
 \begin{equation}\label{eqn:comp2}
     l_{n+1}(a,0)=\frac{a+\EE\left(X_{(n+1)}\right)}{2}\geq \frac{a+\EE\left(c_{(n+1)}\right)}{2}, 
 \end{equation}
 where the equality follows from using the definition of $l_{n+1}(a,0)$ when $a>\EE\left(X_{(n+1)}\right)$ and the inequality hols because $\EE\left(X_{(n+1)}\right)\geq c_{n+1}$. 
 
 Therefore, we obtain $l_{n+1}(a,0) \geq \beta_{n+1}(a)$ for all  $a>\EE\left(X_{(n+1)}\right)$ combining \eqref{eqn:comp1} and \eqref{eqn:comp2}.
 
 \paragraph{Case 2.} Assume that $a \leq \EE\left(X_{(n+1)}\right).$ To prove that  $l_{n+1}(a,0) \geq \beta_{n+1}(a)$, we show that  $l_{n+1}(a,0) \geq \beta_{n+1}$ and that  $l_{n+1}(a,0) \geq (a+c_{n+1})$/2. The latter holds by Lemma~\ref{lem:bound:SPEP}. On the other hand, 
 \begin{equation*}
     l_{n+1}(a,0)=\EE_X(l_n(a \vee X, a \wedge X)) \geq \EE_X(l_n( X,0)) \geq  \EE_X( \beta_n(X))= \beta_{n+1},
 \end{equation*}
 where the first equality follows from the definition of $l_{n+1}(a,0)$ for $a \leq \EE\left(X_{(n+1)}\right)$, the first inequality holds due to the monotonicity of the function $l_n$ in both components, the second inequality follows from the induction hypothesis and the last equality from the definition of $\beta_{n+1}.$
 
 Therefore, we conclude that $l_{n+1}(a,0) \geq \beta_{n+1}(a)$ for all  $a \leq \EE\left(X_{(n+1)}\right)$.
 
 Putting all together we obtain that $l_{n}\geq \beta_{n}$ for all $n$ and the desired result follows. 
 \end{proof}

\subsection{Omitted proofs from Section~\ref{sec:eff}}

The main result in Section~\ref{sec:eff} we prove here is Theorem~\ref{prop:boundeff} and then we work on the competitive selection problem with no recall and two arrivals. 

Before going to the proof, we obtain an expression for the price of anarchy and price of stability if the random variables are distributed according to $F$, namely $\PA_2^{NR}(F)$ and $\PS_2^{NR}(F)$ respectively.

Regarding the price of stability, we need to compute
\[
\frac{ \EE\left(X_{(1:2)}+X_{(2:2)}\right)}{2 \beta_2},
\]
where $2\beta_2$ is obtained from equation \eqref{eq:beta} by taking $n=1$.
That is:
 \begin{eqnarray*}
 2\beta_{2}&=& \int_{0}^{\alpha'_1} 2 \beta_1 \mathrm{d}F(a) + \int_{\alpha'_1}^{\beta_1}  \max\{a+c_1, 2\beta_1 \} \mathrm{d}F(a) +\int_{\beta_1}^{1} (a+c_1) \mathrm{d}F(a). 
 \end{eqnarray*}
 
 Now, due to $E^{NR}_1=\left\{ \left(\frac{\EE(X)}{2},\frac{\EE(X)}{2} \right)\right\}$, we have that $\alpha_1'=\beta_1=\frac{\EE(X)}{2}$. Noting that $c_1= \EE(X)$ and putting all together we have
 
  \begin{eqnarray*}
 2\beta_{2}&=& \int_{0}^{\EE(X)/2} \EE(X) \mathrm{d}F(a) + \int_{\EE(X)/2}^{\EE(X)/2}  \max\{a+\EE(X), \EE(X) \} \mathrm{d}F(a) +\int_{\EE(X)/2}^{1} (a+\EE(X)) \mathrm{d}F(a) \\
 &=&\int_{0}^{1} \EE(X) \mathrm{d}F(a) + \int_{\EE(X)/2}^{1} a \mathrm{d}F(a) \\
 &=& \EE(X) + \PP(X \geq \EE(X)/2) \EE(X | X \geq \EE(X)/2).
 \end{eqnarray*}

On the other hand, $\EE\left(X_{(1:2)}+X_{(2:2)}\right)= 2 \EE(X),$ and therefore

\begin{eqnarray}\label{eqn:pos2}
\frac{1}{\PS_2^{NR}(F)}= \frac{\EE(X) + \PP(X \geq \EE(X)/2) \EE(X | X \geq \EE(X)/2)}{2 \EE(X)}=\frac{1}{2}+\frac{\PP(X \geq \EE(X)/2) \EE(X | X \geq \EE(X)/2)}{2 \EE(X)}.
\end{eqnarray}

Regarding the price of anarchy, we need to compute
\[
 \frac{\EE\left(X_{(1:n)}+X_{(2:n)}\right)}{2 \alpha_2},
\]
where $2 \alpha_2$ follows from equation \eqref{eqn:alphan}  by taking $n=1$.
After some algebra, we obtain 
\begin{eqnarray}\label{eqn:alpha2}\nonumber
2 \alpha_2 &=& \EE(X)+ \int_{\EE(X)/2}^{\EE(X)} 2\EE(X)- \frac{\EE(X)^2}{a} \mathrm{d}F(a)+ \int_{\EE(X)}^1 a \mathrm{d}F(a) \\ 
&=& 2 \EE(X) - \int_0^{\EE(X)/2} a \mathrm{d}F(a) -   \int_{\EE(X)/2}^{\EE(X)} a - 2\EE(X) + \frac{\EE(X)^2}{a} \mathrm{d}F(a),
\end{eqnarray}
and thus
\begin{eqnarray}\label{eqn:poa2}
\frac{1}{\PA_2^{NR}(F)}&=&1- \frac{1}{2 \EE(X)} \int_0^{\EE(X)/2} a \mathrm{d}F(a) -  \frac{1}{2 \EE(X)} \int_{\EE(X)/2}^{\EE(X)} a - 2\EE(X) + \frac{\EE(X)^2}{a} \mathrm{d}F(a).
\end{eqnarray}

Also, we can write the inverse of price of anarchy as follows:
\begin{eqnarray}\label{eqn:poa2_2}
\frac{1}{\PA_2^{NR}(F)}&=&\frac{1}{\PS_2^{NR}(F)}- \frac{1}{2 \EE(X)} \int_{\EE(X)/2}^{\EE(X)} (a-2\EE(X)+ \EE(X)^2/2) \mathrm{d}F(a).
\end{eqnarray}

We now prove Theorem~\ref{prop:boundeff} using the formulas obtained above. 


\begin{proof}[Proof of Theorem~\ref{prop:boundeff}.]
To prove that $\PS_2^{NR}(F) \leq 4/3,$ let us consider the second term in the rhs of \eqref{eqn:pos2} and notice that
\begin{eqnarray*}
\PP(X \geq \EE(X)/2) \EE(X | X \geq \EE(X)/2)&=&\EE(X)-\PP(X < \EE(X)/2) \EE(X | X < \EE(X)/2) \\
&>&\EE(X)-\PP(X < \EE(X)/2)\frac{\EE(X)}{2} \\
&\geq& \EE(X)-\frac{\EE(X)}{2} = \frac{\EE(X)}{2}.
\end{eqnarray*}
 Then, 
\begin{eqnarray*}
\frac{1}{\PS_2^{NR}(F)} \geq \frac{1}{2}+ \frac{\EE(X)}{2} \frac{1}{2 \EE(X)}= \frac{3}{4}.
\end{eqnarray*}

Regarding the price of anarchy, by equation \eqref{eqn:poa2} and using that $ a - 2\EE(X) + \frac{\EE(X)^2}{a} \leq \frac{\EE(X)}{2}$ if $a \in [\EE(X)/2,\EE(X)]$,  we have 
\begin{eqnarray*}
\frac{1}{\PA_2^{NR}(F)}&\geq&1- \frac{\PP(X < \EE(X)/2) \EE(X | X < \EE(X)/2)}{2 \EE(X)} -  \frac{1}{4}\PP(\EE(X)/2 \leq X \leq \EE(X)) .
\end{eqnarray*}

But, $\EE(X | X < \EE(X)/2) \leq \EE(X)/2$ and thus
\begin{eqnarray*}
\frac{1}{\PA_2^{NR}(F)}&\geq&1- \frac{1}{4} \PP(X < \EE(X)/2)  -  \frac{1}{4}\PP(\EE(X)/2 \leq X \leq \EE(X)) \\
&=& 1- \frac{1}{4}  \PP(X \leq \EE(X)) \geq \frac{3}{4} ,
\end{eqnarray*}
obtaining the desire inequalities. 

To prove the tightness of the bound, let us take $\varepsilon > 0$ and $\eta>0$ two small positive real numbers and consider the random variable $X_{\varepsilon, \eta}=(1-\eta) X + \eta Unif[0,1]$, where

  \[   
X= 
     \begin{cases}
     \varepsilon-\varepsilon^2   & \text{with probability } 1-\varepsilon,\\
      1 &\text{with probability } \varepsilon.\\ 
     \end{cases}
\]

Note that when $\eta$ goes to $0$, $$\frac{1}{\PS_2^{NR}(F_{\varepsilon,\eta})} \to \frac{1}{2}+\frac{\PP(X \geq \EE(X)/2) \EE(X | X \geq \EE(X)/2)}{2 \EE(X)}, $$

where $F_{\varepsilon, \eta}$ represents the c.d.f. of $X_{\varepsilon, \eta}$, and therefore it is enough to prove that 
$$\frac{1}{2}+\frac{\PP(X \geq \EE(X)/2) \EE(X | X \geq \EE(X)/2)}{2 \EE(X)} \to 3/4$$
when $\varepsilon$ goes to $0$.

The c.d.f. of $X$ is given by
  \[   
F_\varepsilon(x) = 
     \begin{cases}
       0 & \text{if } x <  \varepsilon-\varepsilon^2  ,\\
      1-\varepsilon &\text{if } x \in [\varepsilon-\varepsilon^2,1),\\ 
      1 & \text{if } x \geq 1; \\
     \end{cases}
\]
and the expected value is $\EE(X)=(\varepsilon-\varepsilon^2)(1-\varepsilon)+\varepsilon=\varepsilon (1+(1-\varepsilon)^2).$

Therefore, after some algebra, it follows that 
\[
\frac{1}{2}+\frac{\PP(X \geq \EE(X)/2) \EE(X | X \geq \EE(X)/2)}{2 \EE(X)}= \frac{1}{2}+ \frac{\varepsilon}{2 \varepsilon (1+(1-\varepsilon)^2)}=\frac{1}{2}+ \frac{1}{2(1+(1-\varepsilon)^2)},
\]
which converges to $3/4$ when $\varepsilon \to 0$, and therefore we obtain price of stability $4/3.$


On the other hand, note that by definition, for each distribution $F$, holds that $\PS_2(F) \leq \PA_2(F)$, and then $\PA_2(F_{\varepsilon,\eta}) \geq 4/3$, where $4/3$ is the upper bound we already prove. Therefore, the bound is tight also for the price of anarchy. 
\end{proof}

\bibliographystyle{abbrv} 
\bibliography{GPR-biblio}

  \end{document}